%% file: arxiv.tex
\documentclass[11pt]{article}

\usepackage{PRIMEarxiv}

\usepackage[utf8]{inputenc} 
\usepackage{booktabs}       
\usepackage{amsfonts}       
\usepackage{nicefrac}       
\usepackage{microtype}      
\usepackage{lipsum}
\usepackage{fancyhdr}       
\usepackage{graphicx}       
\usepackage{arxivpackages}
\usepackage{mymacros}
\usepackage{commands}
\usepackage{natbib}
\renewcommand*{\backrefalt}[4]{%
    \ifcase #1 \footnotesize{(Not cited.)}%
    \or        \footnotesize{(Cited on page~#2.)}%
    \else      \footnotesize{(Cited on pages~#2.)}%
    \fi}

\renewcommand{\cite}[1]{\citep{#1}}
\setcitestyle{authoryear,round,citesep={;},aysep={,},yysep={;}}

\graphicspath{{media/}}     

\newtheorem{assumption}{Assumption}
\newtheorem{remark}{Remark}
\newtheorem{lemma}{Lemma}

\newcount\Comments  
\newcommand{\kibitz}[2]{\ifnum\Comments=1{\color{#1}{#2}}\fi}

\usepackage{color}
\definecolor{english}{rgb}{0.0, 0.5, 0.0}

\title{Generative Adversarial Equilibrium Solvers
}

\author{
  Denizalp Goktas\thanks{Research conducted while the author was an intern at DeepMind.} \\
  Computer Science Department \\
  Brown University \\
  Providence, RI, USA\\
  \texttt{denizalp\_goktas@brown.edu} \\
   \And
  David C. Parkes\thanks{Also, School of Engineering and Applied Sciences, Harvard University.},\ Ian Gemp, Luke Marris, Georgios Piliouras, Romuald Elie, Guy Lever, Andrea Tacchetti \\
  DeepMind Technologies \\
  London, UK\\
  \texttt{\{parkesd, imgemp, marris, gpil, relie, guylever, atacchet\}@deepmind.com} \\
}

\begin{document}
\maketitle
\input{abstract}
\input{intro}

\input{prelim}

\input{learning}
\input{training}
\input{generalization}

\input{experiments}
\input{conclusion}

\bibliographystyle{icml2023}
\bibliography{references}
\newpage
\appendix
\input{appendix}

\end{document}

%% file: abstract.tex
\begin{abstract}
  We introduce the use of generative adversarial learning to compute equilibria in general game-theoretic settings, specifically the \emph{generalized Nash equilibrium} (GNE) in \emph{pseudo-games}, and its specific instantiation as the \emph{competitive equilibrium} (CE) in Arrow-Debreu competitive economies. Pseudo-games are a generalization of games in which players' actions affect not only the payoffs of other players but also their feasible action spaces.
  Although the computation of GNE and CE is intractable in the worst-case, i.e., PPAD-hard, in practice, many applications only require solutions with high accuracy in expectation
  over a distribution of problem instances. We introduce {\em Generative Adversarial Equilibrium Solvers} (\nees{}): a family of generative adversarial neural networks that can  learn GNE and CE from only a sample of problem instances. We provide computational and sample complexity bounds, and apply the framework to finding Nash equilibria in normal-form games, CE in Arrow-Debreu competitive economies, and GNE in an environmental economic model of the Kyoto mechanism.
\end{abstract}

%% file: intro.tex
\section{Introduction}

Economic models and  equilibrium concepts 
are critical tools to solve practical problems, including capacity allocation in wireless and network communication \cite{han2011wireless, pang2008distributed}, energy resource allocation~\cite{hobbs2007power, wei1999power}, and cloud computing~\cite{gutman2012fair, lai2005tycoon, zahedi2018amdahl, ardagna2017cloud}. 
%
%
Many of these economic models 
are instances of what are known as \mydef{pseudo-games}, in which the actions taken by each player affect not only the other players' payoffs, as in games, but also the other players' strategy sets.\footnote{In many games, such as chess, the action taken by one player affects the actions available to the others, but these games are sequential, while in pseudo-games actions are chosen simultaneously. %
Additionally, even if one constructs a game with payoffs that penalize the players for actions that are not allowed, the NE of the ensuing game will in general not correspond to the GNE of the original pseudo-game and can often be trivial. We refer the reader to \Cref{sec_app:gne_vs_ne} for a mathematical example.} 
The formalism of pseudo-games was introduced by \citet{arrow-debreu}, who used
it in studying their foundational microeconomic equilibrium model, the competitive economy model.


%
The standard solution concept for pseudo-games is the {\em generalized Nash equilibrium (GNE)} \cite{arrow-debreu, facchinei2010generalized}, which is an action profile from which no player can improve their payoff by unilaterally deviating to another action in the space of admissible actions determined by the actions of other players. Important economic models can often be formulated as a pseudo-game, with their set of solutions equal to the set of GNE of the pseudo-game: for instance, the set of {
\em competitive equilibria (CE)}~\cite{walras,arrow-debreu} of an Arrow-Debreu competitive economy corresponds to the set of GNE of an 
associated pseudo-game.

A large literature has been devoted to the computation of GNE in certain classes of pseudo-games but unfortunately many algorithms that are guaranteed to converge in theory have in practice been observed to converge slowly in ill-conditioned or large problems or fail numerically~\cite{facchinei2010penalty, jordan2022gne, goktas2022exploitability}. Additionally, all known algorithms have hyperparameters that have to be optimized individually for every pseudo-game instance \cite{facchinei2010generalized}, deteriorating the performance of these algorithms when used to solve multiple pseudo-games. These issues point to a need to develop methods to compute GNE, for a distribution of pseudo-games, reliably and quickly.

We reformulate the problem of computing GNE in pseudo-games (and CE in Arrow-Debreu competitive economies) as a learning problem for a generative adversarial network (GAN) called the \mydef{Generative Adversarial Equilibrium Solver (\nees{})},
consisting of a generator and a discriminator network. 
The generator takes as input a parametric representation of a pseudo-game, and predicts a solution that
consists of a tuple of actions, one per player.
The discriminator takes as input both the pseudo-game and the output of the generator, and outputs a best-response for each player,
seeking to find a useful unilateral deviation for all players; this also gives  the sum of regrets,
with which to evaluate the generator (see \Cref{fig:network_diagram}). \nees{} predicts GNE and CE in batches and in order to minimize the expected exploitability, across a distribution of pseudo-games. \nees{} amortizes computational cost up-front in training, and allows for near constant evaluation time for inference. Our approach is inspired by previous methods that cast the computation of an equilibrium in normal-form games as an unsupervised learning problem~\cite{duan2021pac,marris2022cce}.
These methods train
a network to predict a strategy profile that minimizes  the \mydef{exploitability}  (i.e., the sum of the players' payoff-maximizing unilateral deviations w.r.t.\ a given strategy profile)
over a distribution of games. These methods become inefficient in pseudo-games, since in contrast to regular games, the exploitability in pseudo-games (1) requires solving a non-linear optimization problem, (2) is not Lipschitz-continuous, in turn making it hard to learn from samples, and (3) has unbounded gradients, which might lead to exploding gradients in neighborhoods of GNE. Our GAN formulation  circumvents all three of these issues.

Although the computation of GNE is intractable in the worst-case \cite{chen2006settling, daskalakis2009complexity, chen2009spending, vazirani2011market, garg2017settling}, in practice,  applications may only require a solver that gives solutions with high accuracy in
expectation over a realistic distribution of problem instances. 
In particular, a decision maker may need  to compute a GNE for a sequence of  pseudo-games from some family 
or \emph{en masse} over a 
set of pseudo-games 
sampled from some distribution of interest. 
An example of such an application is the problem of resource allocation on cloud computing platforms \cite{hindman2011mesos,isard2009quincy, burns2016borg,vavilapalli2013apache} where a significant number of methods make use of repeated computation of competitive equilibrium \cite{gutman2012fair, lai2005tycoon, budish2011combinatorial, zahedi2018amdahl, varian1973equity} and generalized Nash equilibrium \cite{ardagna2017cloud, ardagna2011cloud, ardagna2011flexible, anselmi2014generalized}. In such settings, as consumers request resources from the platform, the platforms have to find a new equilibrium  while handling all numerical failures within a given time frame.
Another example is policy makers who often want to understand the equilibria induced by a policy for different distributions of agent preferences in a pseudo-game %
allowing them to study the impact of a policy for a distribution on different possible kinds of participants. For example,
in studying the impact of a protocol such
as the \mydef{Kyoto joint implementation mechanism} (see \Cref{sec:experiments}), one might be interested in understanding how the emission levels of countries would change based on their productivity levels \cite{roger2000kyoto}.
%
%
Other applications 
include computing competitive equilibria
in stochastic market environments. 
For example,  
recently proposed algorithms
work through a series of equilibrium problems, each of which has to be solved quickly~\cite{liu2022welfare}.
%

%
\input{contributions}
\input{figures/figure_intro_network}
%
\input{related}

%% file: contributions.tex
\subsection{Contributions}
%

Earlier  approaches~\cite{duan2021pac,marris2022cce} do not extend even to
continuous (non pseudo-)games, since   evaluating the expected exploitability and its gradient over a distribution of pseudo-games requires  solving as many convex programs as examples in the data set. 
Additionally, in pseudo-games, the exploitability is not Lipschitz-continuous, and thus its gradient is unbounded (\Cref{sec_app:examples}), hindering the use of  standard tools to prove sample complexity and convergence bounds, and making training hard due to exploding gradients.
By delegating the task of computing a best-response to a discriminator,   our method circumvents the issue of solving a convex program, yielding a training  problem given by
a min-max optimization problem whose objective is Lipschitz-continuous, for which gradients can be guaranteed to be bounded under standard assumptions on the discriminator and the payoffs of  players.

Our approach also extends the class of (non pseudo-)games that can be solved through
deep learning methods from normal-form games to simultaneous-move continuous action
games, since  the non-linear program involved in the computation of  exploitability  in 
previous methods makes them  inefficient in application to continuous-action games.
We  give polynomial-time convergence guarantees for our training algorithm for the special case of affine generators and affine discriminators (\Cref{thm:convergence_stationary_approx}, \Cref{sec:theory}) and provide generalization bounds for arbitrary function approximators (\Cref{thm:sample_complexity}, \Cref{sec:theory}). Finally, we provide empirical evidence that \nees{} outperforms state of the art baselines in Arrow-Debreu competitive economies, and show that \nees{} can replicate existing qualitative analyses for pseudo-games, suggesting that \nees{} makes predictions that not only have low expected exploitability, but also are qualitatively correct, i.e., close to the true GNE in action space (\Cref{sec:experiments}).

%% file: figures/figure_intro_network.tex
\begin{figure}
    \centering
    \includegraphics[width=1\linewidth]{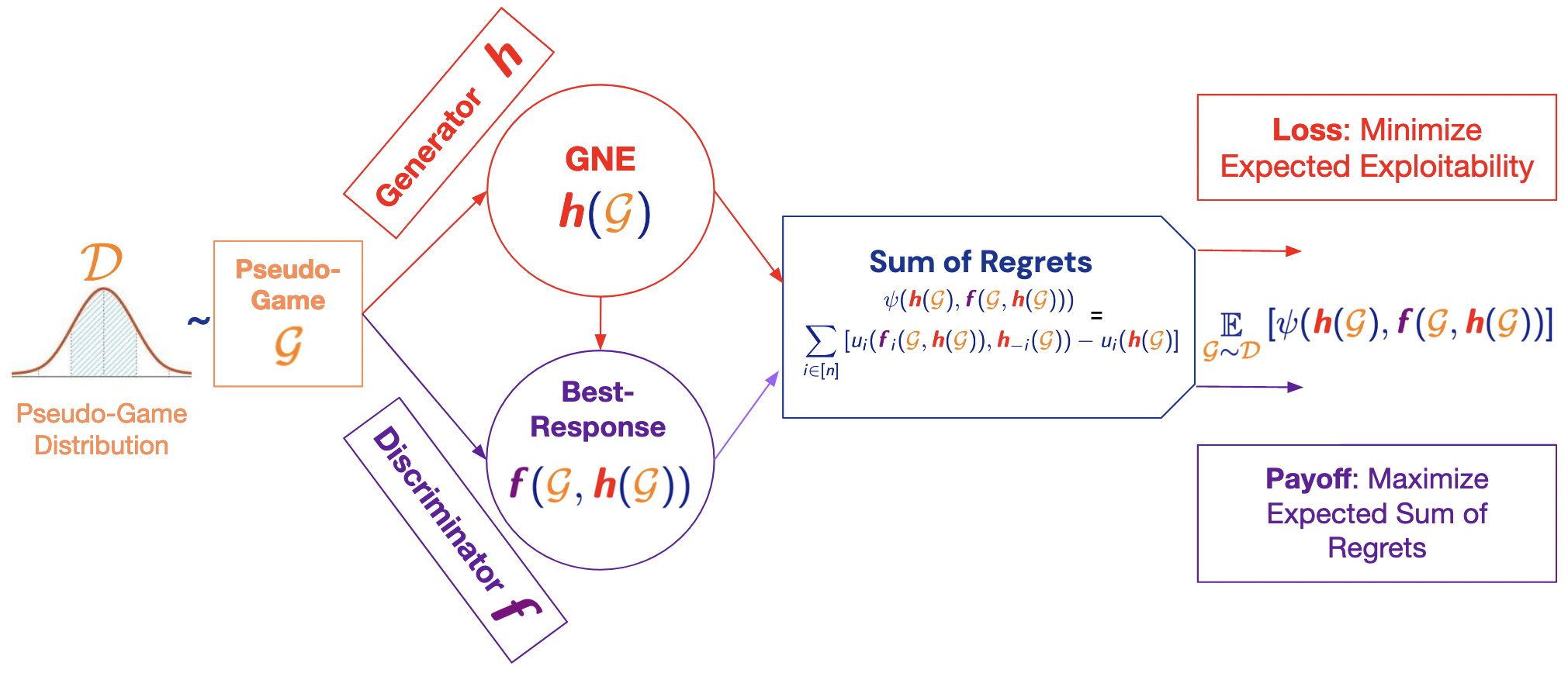}
    \caption{Summary of the Architecture of \nees{}.\label{fig:network_diagram} }
    \vspace{-1em}
\end{figure}

%% file: related.tex
\subsection{Additional Related Work}
We refer the reader to \Cref{sec_app:related} for a survey of methods to compute GNEs, and to \Cref{sec_app:application_motivation} for a survey of applications of GNE. 
Our contributions generally relate to a line of work on \emph{differentiable economics}, which seeks to use methods of neural computation for problems of economic design and equilibrium computation.
In regard to finding optimal economic designs, deep learning has been used  for problems of auction design
%
\cite{dutting2019optimal,curry2022learning,tacchetti2019neural,curry2022differentiable, gemp2022designing, rahme2020auction}
and matching~\cite{ravindranath2021deep}.
In regard to solving for equilibria, some recent works have tried to solve for Nash equilibria in auctions \cite{heidekruger2019computing, bichler2021learning}, and dynamic stochastic general equilibrium models \cite{curry2022finding, chen2021deep, hill2021solving}. 

%% file: prelim.tex
\section{Preliminaries}\label{sec:prelim}
%
%
\subsection{Notation}
All notation for variable types, e.g., vectors, are clear from context, if any confusions arise see \Cref{sec_ap:prelims}.
%
%
We denote the set of integers $\left\{0, \hdots, n-1\right\}$ by $[n]$, the set of natural numbers by $\N$, the set of real numbers by $\R$, and the positive and strictly positive elements of a set by a subscript $+$ and $++$, e.g., $\R_+$ and $\R_{++}$.
%
We denote by $\simplex[n] = \{\x \in \R_+^n \mid \sum_{i = 1}^n x_i = 1\}$, and by $\simplex(A)$, the set of probability measures on the set $A$.
\subsection{Pseudo-Games}

A \mydef{pseudo-game} \cite{arrow-debreu} $\pgame \doteq (\numplayers, \actionspace,  \actions[][\pgame], \actionconstr[][][\pgame], \util[][][\pgame])$, denoted $(\numplayers, \actionspace,  \actions, \actionconstr, \util)$ when clear from context, comprises $\numplayers \in \N_+$ players, where player $\player \in \players$ chooses an action $\action[\player]$ from a non-empty, compact, and convex \mydef{action space} $\actionspace[\player] \subset \R^\numactions$.
We denote the players' joint action space by $\actionspace = \bigtimes_{\player \in \players} \actionspace[\player] \subset \R^{\numplayers \numactions}$.
Each player $\player \in \players$ aims to maximize their continuous \mydef{payoff},
$\util[\player]: \actionspace \to \R$, which is concave in $\action[\player]$, by choosing a feasible action from a set of actions, $\actions[\player](\naction[\player]) \subseteq \actionspace[\player]$, this depending on the actions $\naction[\player] \in \actionspace[-\player] \subset \R^{(\numplayers-1)\numactions}$ of the other players. Here, $\actions[\player]: \actionspace[-\player] \rightrightarrows  \actionspace[\player]$ is a non-empty, continuous, compact- and convex-valued \mydef{(feasible) action correspondence}. It is this dependence on each others' actions that makes this a pseudo-game, and not just a game.
For convenience, we represent $\actions[\player]$
as $\actions[\player](\naction[\player]) = \{ \action[\player] \in \actionspace[\player] \mid \actionconstr[\player][\numconstr](\action[\player], \naction[\player]) \geq \zeros, \text{ for all } \numconstr \in [\numconstrs]\}$, where for all $\player \in \players$, and $\numconstr \in [\numconstrs]$, $\actionconstr[\player][\numconstr]$ is a continuous and concave function in $\action[\player]$, which defines the constraints.\footnote{This is without loss of generality since any compact convex set can be represented by the superlevel sets of a continuous concave function \cite{rockafellar2009variational}.} 
We denote the \mydef{product (feasible) action correspondence} by $\actions(\action) = \bigtimes_{\player \in \players} \actions[\player](\naction[\player])$, which we note is guaranteed to be non-empty, continuous, and compact-valued, but not necessarily convex-valued. We denote $\actions$ the \mydef{set of jointly feasible strategies}, i.e., $\actions = \{\action  \in \actionspace \mid \actionconstr[\player][\numconstr](\action) \geq \zeros, \forall \player \in \players, \numconstr \in [\numconstrs] \}$. We denote the class of all pseudo-games by $\pgames$.\footnote{A \mydef{game} \cite{nash1950existence} is a pseudo-game where, for all players $\player \in \players$, $\actions[\player]$ is a constant correspondence with value $\actionspace[\player]$. A \mydef{normal-form game} is a game where $\actionspace[\player] = \simplex[\numactions]$ and for all $\player \in \players$, $\util[\player]$ is affine.} 

Given a pseudo-game $\pgame$, a \mydef{generalized Nash equilibrium (GNE)} is strategy profile 
$\action^* \in \actions$, s.t.~for all $\player \in \players$ and $\action[\player] \in \actions[\player](\naction[\player][][][*])$, $\util[\player](\action^*) \geq \util[\player](\action[\player], \naction[\player][][][*])$. 
An \mydef{equilibrium mapping}, $\hypothesis: \pgames \to \actions$ is a mapping that takes as input a pseudo-game $\pgame \in \pgames$ and outputs a GNE, $\hypothesis(\pgames)$, for that game. 

Given a pseudo-game $\pgame$, we define the \mydef{regret} for 
player $\player \in \players$ for  action $\action[\player]$  as compared to another action $\otheraction[\player]$, given the action profile $\action[-\player]$ of other players, as 
%
    $\regret[\player][\pgame](\action[\player], \otheraction[\player]; \naction[\player]) = \util[\player][][\pgame](\otheraction[\player], \naction[\player]) - \util[\player][][\pgame](\action[\player], \naction[\player])$.
Additionally, the \mydef{cumulative regret}, $\cumulregret[][\pgame]: \actionspace \times \actionspace \to \R$
between two action profiles $\action \in \actionspace$ and $\otheraction \in \actionspace$
is given by $\cumulregret(\otheraction; \action) = \sum_{\player \in \players} \regret[\player][\pgame](\action[\player], \otheraction[\player]; \naction[\player])$.
Further, the \mydef{exploitability}  (or \mydef{Niakido-Isoda potential  function} \cite{nikaido1955note}),
$\exploit[][\pgame]: \actionspace \to \R$, of an action profile $\action$ is defined as 
$\exploit[][\pgame](\action) =  \sum_{\player \in \players} \max_{\otheraction[\player] \in \actions[\player](\naction[\player])} \regret[\player][\pgame](\action[\player], \otheraction[\player]; \naction[\player])$.
The max is taken over $\actions[\player](\naction[\player])$,
since a player can only deviate within the set of feasible strategies. Note that an action profile $\action^*$ is a GNE iff $\exploit[][\pgame](\action^*) = 0$.
%
%
\subsection{Mathematical Preliminaries}
For any function $\obj: \calX \to \calY$, we denote its Lipschitz-continuity constant by $\lipschitz[\obj]$.
For two arbitrary sets $\hypotheses, \hypotheses^\prime \subset \calF$, the set  $\hypotheses^\prime$ \mydef{$\coveringnumber$-covers} $\hypotheses$ (w.r.t. some norm $\| \cdot \|$) if for any $\hypothesis \in \hypotheses$ there exists $\hypothesis^\prime \in \hypotheses^\prime$ such that $\left\|\hypothesis - \hypothesis^\prime \right\| \leq \coveringnumber$. 
The \mydef{$\coveringnumber$-covering number}, $\coveringcard(\hypothesis, \coveringnumber)$, of a set $\hypotheses$ is the cardinality of the smallest set $\hypotheses^\prime \subset \calF$ that $\coveringnumber$-covers $\hypotheses$.
A set $\hypotheses$ is said to have a \mydef{bounded covering number}, if for all $\coveringnumber \in \R_+$, we have that the logarithm of its covering number is polynomially bounded in $1/\coveringnumber$, that is $\log(\coveringcard(\hypothesis, \coveringnumber)) \leq \poly(\nicefrac{1}{\coveringnumber})$. Additional background can be found in \Cref{sec_ap:prelims}.

%% file: learning.tex
\section{Generative Adversarial Learning of Equilibrium Mappings}\label{sec:GAN}

In this section, we revisit previous formulations of the problem of learning an equilibrium mapping, discuss the computational difficulties associated with these formulations when used to learn GNE, and introduce our generative adversarial learning formulation. 

As creating a sufficiently diverse sample of (pseudo-game, GNE) pairs, while performing adequate equilibrium selection, is intractable both theoretically and computationally, we forgo of supervised learning methods, and formulate the equilibrium mapping learning problem as an \mydef{unsupervised} learning problem, 
following the approach adopted by~\citet{marris2022cce,zhijian2021ne} for finding Nash equilibria.
Given a hypothesis class $\hypotheses \subseteq \actions[][\pgames]$, and a distribution over pseudo-games $\distribpgames \in \simplex(\pgames)$, \mydef{the unsupervised learning problem} for an equilibrium mapping consists of finding a hypothesis $\hypothesis^* \in \argmin_{\hypothesis \in \hypotheses} \Ex_{\pgame \sim \distribpgames} \left[\ell(\pgame, \hypothesis(\pgame))\right]$ where  $\ell: \pgames \times \actionspace \to \R$ is a loss function that outputs the distance of $\action \in \actionspace$ from a GNE,
such that for any pseudo-game $\pgame \in \pgames$, $\ell(\pgame, \action^*) = 0$ iff $\action^*$ is a GNE of $\pgame$.
In particular, \citet{marris2022cce, zhijian2021ne} suggest to use   exploitability
as the loss function.
However, a number of  issues  arise when trying to minimize the expected exploitability over a distribution of pseudo-games:

(1) Computing the gradient of the exploitability, when it exists, for even only one pseudo-game requires solving a concave maximization problem
(this reasoning also applies to continuous games).

(2) The exploitability in pseudo-games, is in general not Lipschitz-continuous (unlike in regular games), even when payoffs are Lipschitz-continuous, since the inputs of the exploitability parameterize the constraints in the optimization problem defining each player's maximal regret computation.
This makes it unclear how
to efficiently approximate $\Ex_{\pgame \sim \distribpgames} \left[\exploit[][\pgame](\action) \right]$ from samples, without knowledge of the distribution.

%

(3) The exploitability in pseudo-games  is absolutely continuous and hence differentiable almost everywhere \cite{afriat1971envelope}, but in contrast to games, the gradients   cannot be bounded. This in turn precludes the convergence of first-order methods.\footnote{
We refer the reader to \Cref{sec_app:examples} for an example in which exploitability is not Lipschitz-continuous and has unbounded gradients.}

To address the aforementioned issues, we propose a generative adversarial learning formulation of the associated 
unsupervised learning problem for equilibrium mappings.
The formulation relies on the following observation, whose proof is deferred to \Cref{sec_app:proofs}: the exploitability can be computed \emph{ex post} after computing the expected cumulative regret by optimizing over the space of best-response functions from pseudo-games to actions, 
rather than the space of actions individually for every pseudo-game.
\begin{restatable}{observation}{obsganeq}\label{obs:gan_eq}
For any $\distribpgames \in \simplex(\pgames)$, we have:
\begin{align}
    &\min_{\hypothesis \in \actions[][\pgames]} \Ex_{\pgame \sim \distribpgames} \left[\exploit[][\pgame](\hypothesis(\pgame)) \right]  \notag \\
    &= \min_{\hypothesis \in \actions[][\pgames]} \max_{\substack{\otherhypothesis \in \actionspace^{\pgames}: \forall \pgame \in \pgames,\\ \otherhypothesis(\pgame) \in \actions[][\pgame](\hypothesis(\pgame))}}  \Ex_{\pgame \sim \distribpgames} \left[ \cumulregret[][\pgame](\hypothesis(\pgame), \otherhypothesis(\pgame))\right].\label{eq:gan_obs}
\end{align}
\end{restatable}
\noindent
By Arrow-Debreu's lemma on abstract economies~\cite{arrow-debreu}, $\hypothesis^*$ is guaranteed to exist and is an equilibrium mapping iff  $\hypothesis^* \in  \argmin\limits_{\hypothesis \in \actions[][\pgames]} \max\limits_{\substack{\otherhypothesis \in \actionspace^{\pgames \times \actions}: \forall \pgame \in \pgames, \\ \otherhypothesis(\pgame, \hypothesis(\pgame)) \in \actions[][\pgame](\hypothesis(\pgame))}}  \Ex_{\pgame \sim \distribpgames} \left[ \cumulregret[][\pgame](\hypothesis(\pgame), \otherhypothesis(\pgame))\right]$.

This  problem formulation allows us to overcome issues (1) and (2).  For (1), rather than solve 
a concave program to compute the exploitability for each pseudo-game and action profile, 
we can learn a function that maps action profiles to their associated best-response profiles (see for example~\citet{lanctot2017unified} for training best-response oracles).
For (2),  the objective function in \Cref{eq:gan_obs} is Lipschitz-continuous when payoff functions are,  which  opens the doors to use standard proof techniques to learn the objective from a polynomial sample of pseudo-games. 

Still, the gradient of $\max_{\otherhypothesis}\Ex_{\pgame \sim \distribpgames} \left[ \cumulregret[][\pgame](\hypothesis(\pgame), \otherhypothesis(\pgame; \hypothesis(\pgame)))\right]$ with respect to $\hypothesis$ is  in general unbounded even when it exists,
due to the constraint $\forall \pgame \in \pgames, \otherhypothesis(\pgame) \in \actions(\hypothesis(\pgame))$. However, since any solution $\otherhypothesis^*(\pgame, \hypothesis(\pgame))$ to the inner optimization problem $\max_{\otherhypothesis \in \actionspace^\pgames: \forall \pgame \in \pgames, \otherhypothesis(\pgame) \in \actions(\hypothesis(\pgame))}  \Ex\left[ \cumulregret[][\pgame](\hypothesis(\pgame), \otherhypothesis(\pgame))\right]$ is implicitly parameterized by the choice of equilibrium mapping $\hypothesis$, we can represent this dependence explicitly in the optimization problem, and restrict our selection of $\otherhypothesis$ to a continuously differentiable hypothesis class $\otherhypotheses \subset \actionspace^{\pgames \times \actions}$, and 
overcome issue (3).

With these observations in mind, given hypothesis classes $\hypotheses \subset \actions[][\pgames]$,
and $\otherhypotheses \subset \actionspace^{\pgames \times \actions}$, the \mydef{generative adversarial learning problem} is to find a tuple $(\hypothesis^*, \otherhypothesis^*) \in \hypotheses \times \otherhypotheses$ that consists  of a \mydef{generator} and \mydef{discriminator} to solve the following optimization problem:
\begin{align}\label{eq:min_max_hypo_abstract}
\!\!    \!\! \min_{\hypothesis \in \hypotheses} \!\!\!\!\max_{\substack{\otherhypothesis \in \otherhypotheses: \forall \pgame \in \pgames,\\ \otherhypothesis(\pgame; \hypothesis(\pgame)) \in \actions[][\pgame](\hypothesis(\pgame))}} \!\!\!\! \Ex_{\pgame \sim \distribpgames} \!\left[ \cumulregret[][\pgame](\hypothesis(\pgame), \!\otherhypothesis(\pgame; \!\hypothesis(\pgame)))\right].
\end{align}
\noindent This problem can be interpreted as a zero-sum game between the generator and the discriminator. The generator takes as input a parametric representation of a pseudo-game, and predicts a solution that consists of an \mydef{action profile}, i.e., a tuple of actions, one per agent. The discriminator  takes the game and the output of the generator as input, and outputs a best-response for each agent (\Cref{fig:network_diagram}). The optimal mappings  $(\hypothesis^*, \otherhypothesis^*)$ for  \Cref{eq:min_max_hypo_abstract} are then called the \mydef{Generative Adversarial Equilibrium Solver (\nees)}.

%% file: training.tex
\section{Convergence and Sample Complexity}\label{sec:theory}


For training, we propose a stochastic variant of the nested gradient descent ascent algorithm \cite{goktas2021minmax}, which we call \mydef{stochastic exploitability descent}. Our algorithm computes the optimal generator and discriminator  by estimating the gradient of the expected cumulative regret  and exploitability on  a training set  of
pseudo-games. 
%
\begin{algorithm}[!ht]
\caption{Stochastic Exploitability Descent}
\textbf{Inputs:} $\batchpgames, \learnrate[\hypothesis], \learnrate[\otherhypothesis], \numiters[\hypothesis], \numiters[\otherhypothesis],
\weight[][\hypothesis][0],
\weight[][\otherhypothesis][0]$\\
\textbf{Outputs:} $(\weight[][\hypothesis][\iter],
\weight[][\otherhypothesis][\iter])_{\iter = 0}^{\numiters[\hypothesis]}$
\label{alg:cumul_regret}
\begin{algorithmic}[1]
\FOR{$\iter = 0, \hdots, \numiters[\hypothesis] - 1$}
    \STATE Receive batch $\batchpgames[][\iter] \subset \samplepgames$.
    \STATE
    $\weight[][{\hypothesis}][\iter+1] = \weight[][{\hypothesis}][\iter] - \learnrate[\hypothesis][\iter] \left( \nicefrac{1}{\left|\batchpgames[\hypothesis][\iter] \right|}  \sum_{\pgame \in \batchpgames[\hypothesis][\iter]} \left[ \grad[{\weight[][\hypothesis]}] \avg[\cumulregret](\weight[][\hypothesis][\iter], \weight[][\otherhypothesis][\iter]) \right]\right)$
    \STATE $\weight[][\otherhypothesis]  = \zeros$
    \FOR{$s = 0, \hdots, \numiters[\otherhypothesis] - 1$}
        \STATE Receive batch $\batchpgames[][s] \subset \samplepgames$.
            \STATE $\weight[][\otherhypothesis]  =  \weight[][\otherhypothesis] + \learnrate[\hypothesis][s] \left( \nicefrac{1}{\left| \batchpgames[\otherhypothesis][s] \right|}  \sum_{\pgame \in \batchpgames[\otherhypothesis][s]} \grad[{\weight[][\otherhypothesis]}] \avg[\cumulregret](\weight[][\hypothesis][\iter], \weight[][\otherhypothesis]) \right)$
    \ENDFOR
    \STATE $\weight[][\otherhypothesis][\iter + 1] =  \weight[][\otherhypothesis]$
    \ENDFOR
\STATE Return $(\weight[][\hypothesis][\iter],
\weight[][\otherhypothesis][\iter])_{\iter = 0}^{\numiters[\hypothesis]}$
\end{algorithmic}
\end{algorithm}

\subsection{Training Algorithm}
For purposes of applicability, going forward, we will assume that we have access to the distribution of pseudo-games $\distribpgames$ only indirectly through a \mydef{training set} $\samplepgames \sim \distribpgames$ of $\numsamples \in \N_+$ sampled pseudo-games.
 %
%
Additionally, we will assume that the generator $\hypothesis \in \hypotheses$ and discriminator $\otherhypothesis \in \otherhypotheses$ are parameterized by  vectors, $\weight[][\hypothesis], \weight[][\otherhypothesis] \in \R^\numweights$, such that for all pseudo-games $\pgame \in \pgames$, and weight vectors $\weight \in \R^\numweights$, $\hypothesis(\pgame; \weight[][\hypothesis]) \in \actions$ and $\otherhypothesis(\pgame; \hypothesis(\pgame; \weight[][\hypothesis]), \weight[][\otherhypothesis]) \in \actions(\hypothesis(\pgame; \weight[][\hypothesis]))$.

For notational simplicity,  we define the \mydef{expected exploitability} and \mydef{expected cumulative regret}, respectively, as:
\begin{align*}
    &\mean[\cumulregret](\weight[][\hypothesis], \weight[][\otherhypothesis]) = \Ex \left[ \cumulregret[][\pgame](\hypothesis(\pgame; \weight[][\hypothesis]), \otherhypothesis(\pgame, \hypothesis(\pgame; \weight[][\hypothesis]); \weight[][\otherhypothesis]))\right]\\
    &\mean[\exploit](\weight[][\hypothesis]) = \max_{\weight[][\otherhypothesis] \in \R^\numweights: \forall \pgame \in \pgames, \otherhypothesis(\pgame; \weight[][\otherhypothesis]) \in \actions(\hypothesis(\pgame, \weight[][\hypothesis]))} \mean[\cumulregret](\weight[][\hypothesis], \weight[][\otherhypothesis]),
\end{align*}
\noindent
where the expectation is over the distribution of pseudo-games $\pgame \sim \distribpgames$.
Similarly, we define the \mydef{empirical cumulative regret} and the \mydef{empirical exploitability} respectively as:
\begin{align*}
    &\avg[\cumulregret](\weight[][\hypothesis], \weight[][\otherhypothesis]) = \Ex \left[ \cumulregret[][\pgame](\hypothesis(\pgame; \weight[][\hypothesis]), \otherhypothesis(\pgame; \hypothesis(\pgame; \weight[][\hypothesis]), \weight[][\otherhypothesis]))\right]\\
    &\avg[\exploit](\weight[][\hypothesis]) = \max_{\weight[][\otherhypothesis] \in \R^\numweights: \forall \pgame \in \pgames, \otherhypothesis(\pgame; \weight[][\otherhypothesis]) \in \actions(\hypothesis(\pgame; \weight[][\hypothesis]))} \avg[\cumulregret](\weight[][\hypothesis], \weight[][\otherhypothesis]),
\end{align*}
\noindent
where the expectation is over the uniform distribution over the training set, $\pgame \sim \mathrm{unif}(\samplepgames)$.
Putting this together, our training problem becomes:
\begin{align}
    \min_{\weight[][\hypothesis] \in \R^\numweights} \max_{\substack{\weight[][\otherhypothesis] \in \R^\numweights: \forall \pgame \in \pgames,\\ \otherhypothesis(\pgame; \hypothesis(\pgame; \weight[][\hypothesis]), \weight[][\otherhypothesis]) \in \actions(\hypothesis(\pgame; \weight[][\hypothesis]))}}  \mean[\cumulregret](\weight[][\hypothesis], \weight[][\otherhypothesis]).
\end{align}
We propose \Cref{alg:cumul_regret} to solve this optimization problem. This is
a nested stochastic gradient descent-ascent algorithm, which for each generator descent step runs multiple stochastic gradient ascent steps on the weights of the discriminator to approximate the empirical cumulative regret by processing the pseudo-games in the training set in \mydef{batches}, i.e., as mutually exclusive subsets of the training set. After the stochastic gradient ascent steps are done, the algorithm then takes a step of stochastic gradient descent on the empirical exploitability w.r.t. the weights of the generator using the discriminator's weights computed by the stochastic gradient ascent steps to compute the gradient. Note that when the hypothesis class of the discriminator is assumed to contain only differentiable functions, the gradient of the generator with respect to its weights exist, and they are given by the implicit function theorem. 

\subsection{Convergence bounds.}
We  give assumptions under which our
algorithm converges to a stationary point of the empirical exploitability in polynomial-time.
\begin{assumption}
\label{assum:convergence}
For any player $\player \in \players$ and $\pgame \in \mathrm{supp}(\distribpgames)$:
1.~(Lipschitz-smoothness) their payoff $\util[\player][][\pgame]$ is $\lipschitz[{\grad[\util]}]$-Lipchitz smooth,
2.~(Strong concavity) their payoff $\util[\player][][\pgame]$ is $\scparam$-strongly-concave in $\action[\player]$,
and 3.~(Affinely parameterized hypothesis classes) For all $\hypothesis \in \hypotheses \subset \actions[][\pgames \times \R^\numweights]$, $\otherhypothesis \in \otherhypotheses \subset \actions[][\pgames \times \actions \times \R^\numweights]$, $\pgame \in \pgames$, $\hypothesis(\pgame; \cdot)$, and $\otherhypothesis(\pgame; \cdot)$ are affine. 
\end{assumption}

The following assumption
is implied by \Cref{assum:convergence}.
\begin{assumption}
\label{assum:sample}
(Lipschitz-continuity) For any player $\player \in \players$ and $\pgame \in  \supp(\distribpgames)$, their payoff $\util[\player][][\pgame]$ is $\lipschitz[{\grad[\util]}]$-Lipchitz continuous. 
\end{assumption}

 Lipschitz-smoothness is a standard assumption in the convex optimization literature \cite{boyd2004convex}, while strong concavity in each player's action is an assumption that is much weaker than 
 strong monotonicity of the pseudo-game which is commonly used in the GNE literature to prove convergence results \cite{jordan2022gne}.\footnote{Strong concavity is necessary to obtain our results, since it ensures that the composition of the discriminator and the exploitability satisfy the PL-condition \cite{karimi2016linear} allowing us to obtain convergence to an optimal solution for the discriminator.}
 For omitted definitions, and proofs/results we refer the reader to \Cref{sec_ap:prelims} and \Cref{sec_app:proofs} respectively.

 Theorem~\ref{thm:convergence_stationary_approx}
tells us that our algorithm converges to a stationary point of the empirical exploitability at a  $\tilde{O}(\nicefrac{1}{\sqrt{\numiters[\hypothesis]}})$ rate, up to an error term that depends linearly on the distance, $\varepsilon$, of the discriminator computed by the algorithm w.r.t.~to the optimal discriminator.
A smaller $\varepsilon$ results in higher accuracy, at the expense of a longer run time.
%
\begin{restatable}[Convergence to Stationary Point]{theorem}{thmconvergencestationaryapprox}
\label{thm:convergence_stationary_approx}
Suppose that \Cref{assum:convergence} holds. Let $\singularval[\mathrm{min}](\otherhypothesis)$ be the smallest non-zero singular value of $\otherhypothesis$. If \Cref{alg:cumul_regret} is run with learning rates $\learnrate[\hypothesis][\iter] = \nicefrac{1}{\sqrt{\iter}}$ and $\learnrate[\otherhypothesis][s] = \nicefrac{(2 s + 1)}{(2\singularval[\mathrm{min}](\otherhypothesis)\scparam(s + 1)^2)}$ for all $\iter \in [\numiters[\hypothesis]]$, $s \in [\numiters[\otherhypothesis]]$. Then, for any number of iterations $\numiters[\hypothesis] \in \N_{++}$, and  $\numiters[\otherhypothesis] \geq \frac{\lipschitz[{\grad \avg[\cumulregret]}]^{3} \lipschitz[{\avg[\cumulregret]}]^2 }{2 \singularval[\mathrm{min}](\otherhypothesis)\scparam \varepsilon}$,
the outputs $(\weight[][\hypothesis][\iter])_{t = 0}^{\numiters[\hypothesis]}$
satisfy
$\left( \min_{\iter = 0, \hdots, \numiters[\hypothesis] - 1} \left\| \grad[\action] \avg[\exploit](\weight[][\hypothesis][\iter]) \right\|_2^2 \right) \in O\left(\frac{\log(\numiters[\hypothesis])}{\sqrt{\numiters[\hypothesis]}} + \varepsilon\right)$ where $\varepsilon > 0$.
\end{restatable}

\begin{remark}
This convergence 
can be further improved to a convergence to a minimum of the empirical exploitability under stronger assumptions on the pseudo-games that ensure that the empirical exploitability is convex. For instance, for zero-sum, potential, a large class of monotone games,\footnote{We recall that a pseudo-game is said to be \mydef{monotone} (or \mydef{dissipative}) if for all $\action, \otheraction \in \actionspace$
$\sum_{\player \in \players} \left(\grad[{\action[\player]}] \util[\player](\action) - \grad[{\action[\player]}] \util[\player](\otheraction) \right)^T \left( \action - \otheraction \right) \leq 0$.} and a class of bilinear games, 
the empirical exploitability is guaranteed to be convex \cite{flam1994gne}.
\end{remark}

%% file: generalization.tex
\subsection{Generalization bounds.}
We also give a sample complexity result for the  
empirical exploitability, showing how cumulative regret can be approximated 
with a sample of pseudo-games that is polynomial in the parameters of the game distribution, 
$\nicefrac{1}{\varepsilon}$ and $\nicefrac{1}{\delta}$.
%
The novelty of the result comes from  the  context: 
we mentioned earlier that expected exploitability need not be 
Lipschitz-continuous, making it hard 
to use any standard and simple machinery to prove a  sample complexity
bound.   However, by reframing this problem 
as one of learning the expected cumulative regret, which is  Lipschitz-continuous, we can obtain the result.
%
%
%
%
\begin{restatable}[Sample Complexity of Expected Cumulative Regret]{theorem}{thmsamplecomplexity}\label{thm:sample_complexity}
For any hypothesis classes $\hypotheses \subset \actions[][\pgames]$, $\otherhypotheses \subset \actionspace^{\actions \times \pgames}$ $\varepsilon, \delta \in (0,1)$ and any pseudo-game distribution $\distribpgames \in \simplex(\pgames)$, with probability at least $1-\delta$ over draw of the training set $\samplepgames \sim \distribpgames^\numsamples$ with 
    $\numsamples \geq \frac{9 }{\varepsilon^2} \left[ \log\left(\frac{\coveringcard(\hypotheses, \nicefrac{\varepsilon}{6})}{\delta}\right) + \log\left(\frac{\coveringcard(\otherhypotheses, \nicefrac{\varepsilon}{6})}{\delta}\right) \right]$,
we have:
\begin{align*}
    \left| \mathop{\Ex}_{\pgame \sim \mathrm{unif}(\samplepgames)} \! \left[\cumulregret[][\pgame](\hypothesis(\pgame), \otherhypothesis(\pgame, \hypothesis(\pgame))) \right] \!\!- \!\! \mathop{\Ex}_{\pgame \sim \distribpgames} \! \left[\cumulregret[][\pgame](\hypothesis(\pgame), \otherhypothesis(\pgame, \hypothesis(\pgame))) \right] \right| \leq \varepsilon
\end{align*}
\end{restatable}

%% file: experiments.tex
\section{Experimental Results}\label{sec:experiments}

We run three sets of experiments, in which we train \nees{} in three different pseudo-game settings.\footnote{We include experiments on normal-form games, as well as all missing additional implementation details in \Cref{sec_app:experiments}.} 
%
%
All  experiments are run with 5 randomly selected different seeds, with hyperparameter selection being done over 
all 5 seeds 
and confidence intervals reported across these seeds as appropriate.

\subsection{Arrow-Debreu Exchange Economies}
Our first
set of experiments aim to solve CE in Arrow-Debreu exchange economies \cite{arrow-debreu},
%
a special case of competitive economies with only consumers.\footnote{This is without loss of generality, as Arrow-Debreu competitive economies can be reduced to Arrow-Debreu exchange economies in polynomial-time.}
The difficulty in solving the pseudo-game associated with Arrow-Debreu exchange economies---hereafter \mydef{exchange economies}---arises from the fact that it does not fit into any well-defined categories of pseudo-games, e.g., monotone or jointly convex, for which there are  algorithms that converge to GNEs.

An \mydef{exchange economy} $(\util, \consendow)$ consists of a finite set of $\numgoods \in \N_+$ goods and $\numbuyers \in \N_+$ consumers (or traders). Every consumer $\buyer \in \buyers$ has a set of possible consumptions $\consumptions[\buyer] \subseteq \mathbb{R}^{\numgoods}_+$, an endowment of goods $\consendow[\buyer] = \left(\consendow[\buyer][1], \dots, \consendow[\buyer][\numgoods] \right) \in \R^\numgoods$ and a utility function $\util[\buyer]: \mathbb{R}^{\numgoods } \to \mathbb{R}$. We denote 
$\consendow = \left(\consendow[1], \hdots, \consendow[\numbuyers] \right)^T$. 
Any exchange economy can be formulated as a pseudo-game whose set of GNE is equal to the set of competitive equilibria (CE)\footnote{We refer the reader to \Cref{sec_ap:arrow_debreu} for a definition.} 
of the original economy \cite{arrow-debreu}. This pseudo-game consists of $\numbuyers + 1$ agents, who correspond to the $\numbuyers$ buyers and a seller.
The pseudo-game is given by the following optimization problem, for each buyer $\buyer \in \buyers$, and the seller, respectively:
\begin{align*}
    \!\!\max_{\allocation[\buyer] \in \consumptions[\buyer]: \allocation[\buyer] \cdot \price \leq \consendow[\buyer] \cdot \price} \!\! \util[\buyer](\allocation[\buyer]) \quad \vline \quad \max_{\price \in \simplex[\numgoods]} \price \cdot \!\!\left( \sum_{\buyer \in \buyers} \allocation[\buyer] - \!\!\sum_{\buyer \in \buyers} \consendow[\buyer] \!\!\right).
\end{align*}
Let $\valuation[\buyer] \in \R^\numgoods_+$, $\subsparams \in (-\infty, 1]^\numbuyers$, be a vector of parameters for the utility function of buyer $\buyer \in \buyers$.
In our experiments, we consider the following utility function classes:
\mydef{Linear}: $\util[\buyer](\allocation[\buyer]) = \sum_{\good \in \goods} \valuation[\buyer][\good] \allocation[\buyer][\good]$,  \mydef{Cobb-Douglas}:  $\util[\buyer](\allocation[\buyer]) = \prod_{\good \in \goods} \allocation[\buyer][\good]^{\valuation[\buyer][\good]}$, \mydef{Leontief}:  $\util[\buyer](\allocation[\buyer]) = \min_{\good \in \goods} \left\{ \frac{\allocation[\buyer][\good]}{\valuation[\buyer][\good]}\right\}$, and 
\mydef{constant elasticity of substitution (CES)}: $\util[\buyer](\allocation[\buyer]) = \left(\sum_{\good \in \goods} \valuation[\buyer][\good] \allocation[\buyer][\good]^\rho\right)^{\nicefrac{1}{\rho}}$.  When we take $\rho = 1$, $\rho \to 0$, and $\rho \to -\infty$ for CES utilities, we obtain linear, Cobb-Douglas, and Leontief utilities respectively. We denote $\valuation = (\valuation[1], \hdots, \valuation[\numbuyers])^T$. As is standard in the literature~\cite{fisher-tatonnement, branzei2021proportional}, we assume that for all buyers $\buyer \in \buyers$, $\consumptions[\buyer] = \R_+^\numgoods$. Once a utility function class is fixed, an exchange economy is referred to with the name of the utility function, and can be sampled as a tuple $(\valuation, \consendow) \in \R^{\numbuyers \times \numgoods} \times \R^{\numbuyers \times \numgoods}$ for linear, Cobb-Douglas, and Leontief exchange economies, and as a tuple $(\valuation, \bm{\rho}, \consendow) \in \R^{\numbuyers \times \numgoods} \times \R^{\numbuyers} \times \R^{\numbuyers \times \numgoods}$ for CES exchange economies. For CES exchange economies, 
 we have a \mydef{gross substitute (GS)} and \mydef{gross complements (GC)} CES economy,
 either when $\rho_\buyer \geq 0$ for all buyers 
or  $\rho_\buyer < 0$  for all buyers,  respectively.  Otherwise, this is a
mixed CES economy.

\paragraph{Baselines}
For  special cases of exchange economies, the computation of CE is well-studied (e.g., \citet{bei2015tatonnement}), allowing us to compare the performance of  \nees{} to known specialized methods.
We benchmark \nees{} to \mydef{t\^atonnement} \cite{walras}, which is
an auction-like algorithm that is guaranteed to converge for CES utilities with $\bm{\rho} \in [0, 1)^\numplayers$ \cite{bei2015tatonnement}, including Cobb-Douglas and excluding Linear utilities. We also
benchmark to  \mydef{exploitability descent (or Augmented Descent Ascent)} \cite{goktas2022exploitability}. For each 
of these baselines, we run an extensive grid search over  decreasing learning rates during validation 
(see \Cref{sec_ap:arrow_debreu}). Each baseline was run to convergence.

We report the results of two experiments. First, we 
run our algorithms in linear, Cobb-Douglas, Leontief, GS CES, GC CES, and mixed CES exchange economies (we defer the
results from the GS CES, and  GC CES experiments to the Appendix). We report the distribution of the
exploitability on the test set for \nees{} and the baselines in \Cref{fig:violin_ad} (additional plots can be found in \Cref{sec_ap:arrow_debreu}). 
We measure  performance 
w.r.t. the \mydef{normalized exploitability},
which is  the exploitability of the method divided by the average exploitability over the action space. 
We observe that in all economies, \nees{} outputs an action profile that is on average better than at least $99\%$ of the action profiles in terms of exploitability. 
In all four markets, \nees{} on average achieves lower exploitability than the baselines (see \Cref{fig:arrow_exploit}, \Cref{sec_ap:arrow_debreu}). 
We also see in \Cref{fig:violin_ad} that \nees{} outperforms the baselines, in distribution, in every economy
except Cobb-Douglas. This is not surprising since t\^atonnement is guaranteed to converge in Cobb-Douglas economies \cite{bei2015tatonnement}.
That said,  t\^atonement does not outperform \nees{} on average, since  t\^atonnement's convergence guarantees hold for different learning rates in each market. 

Second, we compare  \nees{} with the performance of t\^atonnment 
in a pathological and well-know Leontief exchange economy, \mydef{the Scarf economy} \cite{scarf1960instable} (\Cref{fig:scarf_economy}). 
Here, we soft start t\^atonnement with the output of \nees{} with some added uniform noise. This additional noise ensures that the starting point of t\^atonnement is distinct from the output of \nees{}).
We see on Figure \ref{fig:scarf_economy} that the prices generated by t\^atonnement  spiral out and settle into an orbit. 
This makes sense as, unlike pure Nash equilibria, GNE and CE are not locally stable \cite{flokas2020no}, meaning that soft-starting the algorithm with the output of \nees{} might be worse than using \nees{} alone.\footnote{An implication is that, whereas  \citet{duan2021pac} observe that neural equilibrium mappings can be used to succesfully soft start known learning dynamics that compute Nash equilibria, this 
does not hold in pseudo-games.} 

The success of \nees{} in Leontief exchange economies as well as the Scarf economy is notable, and suggests that  \nees{} might be smoothing out the loss landscape since for these economies the exploitability is non-differentiable. 
%
%
\begin{figure*}[t]
  \centering
  \begin{subfigure}{0.52\textwidth}
    \includegraphics[width=\linewidth,trim={0cm 0cm 0.25cm 0.2cm},clip]{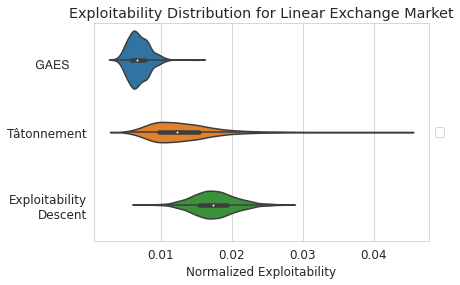}
    \caption{Linear economy}
  \end{subfigure}
  \begin{subfigure}{0.42\textwidth}
    \includegraphics[width=\linewidth,trim={3.1cm 0cm 2cm 0.75cm},clip]{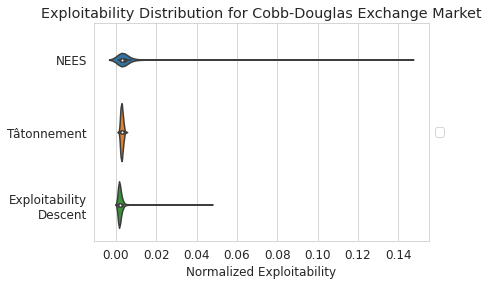}
    \caption{Cobb-Douglas economy}
  \end{subfigure}
  \begin{subfigure}{0.52\textwidth}     
    \includegraphics[width=\linewidth,trim={0cm 0cm 0.3cm 0.2cm},clip]{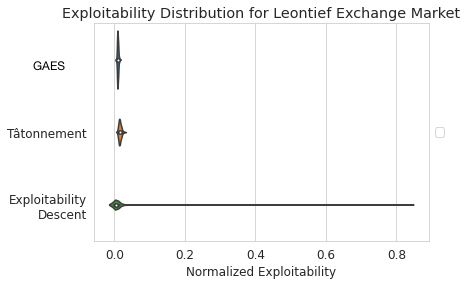}
    \caption{Leontief economy}
  \end{subfigure}
  \begin{subfigure}{0.42\textwidth}
    \includegraphics[width=\linewidth,trim={3.1cm 0cm 1.6cm 0.75cm},clip]{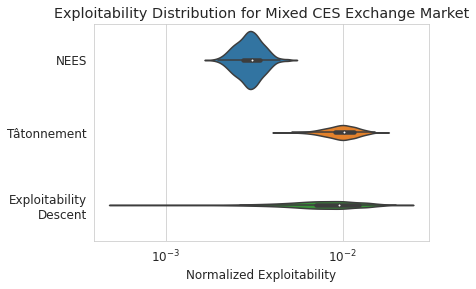}
    \caption{Mixed CES economy}
  \end{subfigure}
  \caption{The distribution of test exploitability on pseudo-games.
   \nees{} outperforms all baselines on average in all markets and in distribution in all markets except Cobb-Douglas.\label{fig:violin_ad} 
   }
\end{figure*}
\begin{figure}[!ht]
    \centering
    \includegraphics[width =0.5\linewidth]{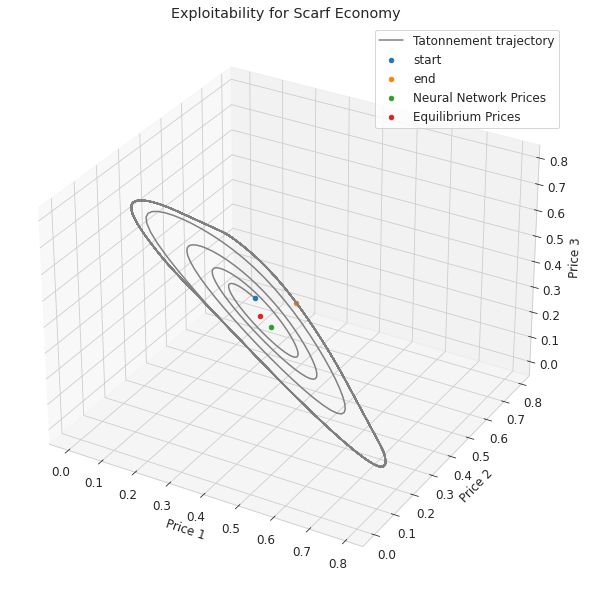}
    \caption{A phase portrait of equilibrium prices in the Scarf Economy.  While the output of \nees{} is close to the equilibrium prices, the final prices outputed by t\^atonnement prices are further than the starting prices. 
    \label{fig:scarf_economy}}
    \vspace{-1em}
\end{figure}

\subsection{Kyoto Joint Implementation Mechanism}
We solve a pseudo-game model of the joint implementation mechanism proposed in the Kyoto protocols~\cite{protocol1997kyoto}. The {\em Kyoto Joint Implementation Mechanism} is a {\em cap-and-trade} mechanism that  bounds each country that signed onto the protocol
to emit anthropogenic gases below a particular emission cap. 
Countries bound by the mechanism can  also invest in green projects in other countries, which in return increases their emission caps. \citet{breton2006game} introduce a model of the Kyoto Joint Implementation Mechanism, using this to predict the impact of the mechanism. The model that the authors propose is partially solvable analytically, that is, one can characterize equilibria qualitatively using comparative statics~\cite{nachbar2002general}, but cannot obtain closed-form solutions for GNE.
Moreover, there is no algorithm that is known to converge to the GNE of this pseudo-game,
since it is not monotone. 

Formally, the \mydef{(Kyoto) Joint Implementation Mechanism} (JI)
consists of $\numcountries \in \N_+$ countries. Each \mydef{country} $\country \in \countries$, can make decisions that result in environmentally damaging anthropogenic \mydef{emissions}, $\emit[\country] \in \R_+$, and can make investments $\invest[\country] \in \R^\numcountries$ to offset their emissions. 
The investments $\invest[\country] \in \R^\numcountries$ that each country $\country \in \countries$ makes in another country $\othercountry \in \countries$ offsets that country's emissions by $\invest[\country][\othercountry] \investrate[\othercountry]$, where this is in proportion to an investment return rate $\gamma_j>0$ for country $j$, with
$\investrate \in \R^\numcountries_+$.
Each country $\country$ has a revenue $\revenue[\country]: \R_+^\numcountries \to \R$, which is a
function of its emissions $\emit[\country]$, a cost function $\cost[\country]: \R^{\numcountries \times \numcountries} \to \R$ that is  a function of all investments,
and a negative externality function, $\damage[\country]: \R^\numcountries \to \R$, 
that is a function of the net emissions of all countries.

Each country $\country$ aims to maximize their surplus, which is equal to their revenue minus the costs of their investments, as well as  the negative externalities caused by emissions constrained by keeping their net emissions under an emission cap, $\emitcap[\country] \in \R_+$. Additionally, for all countries $\country \in \countries$, we require the emission transfer balance to hold, i.e., $\emit[\country] - \sum_{\othercountry \in \countries} \invest[\othercountry][\country] \investrate[\othercountry] \geq 0$, that is no country can transfer more emission reduction than they have. This gives the following optimization problem for each country, $\country \in \countries$:
\begin{align*}
   \!\!\! \max_{(\emit[\country], \invest[\country]) \in \R_+ \times \R_+^{\numcountries}}  &\!\! \revenue[\country](\emit[\player]) \!-\! \cost[\country](\invest) \!-\! \damage[\country]\!\left(\! \left\{\! \emit[\country] \!-\! \sum_{\othercountry \in \countries}\!\invest[\othercountry][\country] \investrate[\country] \!\right\}_\country \right)\\
    & \emit[\country] - \sum_{\othercountry \in \countries} \invest[\country][\othercountry] \investrate[\othercountry] \leq \emitcap[\country]\\
    & \emit[\country] - \investrate[\country] \sum_{\othercountry \in \countries} \invest[\othercountry][\country]  \geq 0 \quad \quad \forall \country  \in \countries
\end{align*}
The literature has traditionally assumed that $\revenue[\country](\emit[\player]; \revrate[\country]) = \emit[\player]\left( \revrate[\player] - \frac{1}{2}\emit[\player] \right)$, $\cost[\country](\invest) = \frac{1}{2} \left( \invest[\country][\country]^2 + \sum_{\othercountry \neq \country} \left[ 
(\invest[\country][\othercountry] + \invest[\othercountry][\othercountry])^2 - \invest[\othercountry][\othercountry]^2 \right] \right)$, and $\damage[\country](\emit - \sum_{\country \in \countries} \invest[\country], \damagerate[\country]) = \damagerate[\country] \sum_{\country \in \countries} (\emit[\country] - \sum_{\othercountry \in \countries} \invest[\country][\othercountry])$.
Fixing these functional forms, we can  sample $(\revrate, \damagerate, \investrate, \emitcap) \in \R^\numagents_+ \times \R^\numagents \times \R^{\numagents \times \numagents} \times \R^\numagents$ and obtain a representation of the JI mechanism.

We run two different experiments to replicate and extend the analysis of~\citet{breton2006game}. 
We first replicate their qualitative analysis of equilibria (\Cref{fig:phase_kyoto}). \citet{breton2006game} introduce a comparative static analysis, in which they fix all parameters of  JI   except for $\revrate$, and characterize the kinds of GNE  as $\revrate$ varies. In \Cref{fig:phase_kyoto}, the six regions correspond to different kinds of GNE. 
For instance in Region~1, both countries emit strictly less than their emission cap (see Section~4 of \citet{breton2006game}). The parts of the plot that are not shaded correspond to pseudo-games for which GNE are not unique,
and for which equilibria cannot be characterized analytically. We superpose on top of this plot
a set of pseudo-games from an unseen test set, and color each pseudo-game by the color of the region whose condition they fulfill.  


We observe that with the exception of Regions~1 and~2a, the structure of the GNE generated by \nees{} 
lines up well with 
the type of GNE. We believe that failure to predict Regions~1 and~2a is
due to  the closeness of the action profiles that fit either of these equilibrium types to other ones.  For instance, \nees{} predicts a GNE of type~4b for pseudo-games in Region~2a but these GNE, although qualitatively different, are very close in action space: the only difference between the two regions is that for the type 2a GNE, one player emits strictly less than its  cap and the other emits exactly its cap, while for the type 4b GNE, both players emit exactly their
emission cap. 
%
A similar conclusion holds for GNE in Region 3a, as predicted in Region~1.

We also solve the JI pseudo-game for a distribution of JI pseudo-games (\Cref{fig:exploit_kyoto}, \Cref{sec_app:experiments}), and 
these results confirm that the testing normalized exploitability is very low. 
We note that \mydef{normalized exploitability} is given as the exploitability 
divided by the average exploitability over the action space, which means that \nees{} has on average a lower exploitability than $\approx 99.5\%$ of the feasible action profiles. 
This confirms our hypothesis that failure to predict Regions~1 and~2a in \Cref{fig:phase_kyoto} arises from
the proximity between GNE of these types and GNE of types 3a or 4a respectively, since according to exploitability there is very little improvement
left for \nees{}.
\begin{figure}[t]
    \centering
    \includegraphics[width=0.6\linewidth]{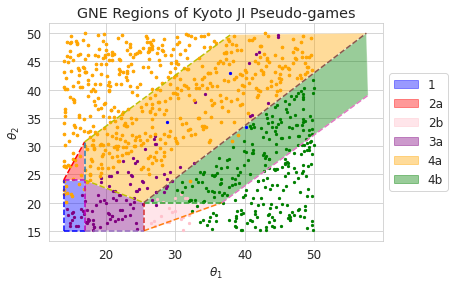}
    \caption{A  taxonomy of different equilibrium types for various  pseudo-games obtained by fixing all parameters, and varying $\revrate$. The $x$ and $y$ axes represent the revenue parameters $\revrate[1]$, $\revrate[2]$ of the countries. The colored regions (obtained analytically) correspond to different qualitative types of GNE  while the dots correspond to  pseudo-games in the test set colored by the GNE types that were predicted for them by \nees{}. 
    \label{fig:phase_kyoto}}
    \vspace{-2em}
\end{figure}

%% file: conclusion.tex
\section{Conclusion}
We introduced \nees{}, a GAN that learns mappings from pseudo-games to GNE. \nees{} outperforms existing methods to compute GNE or CE in exchange economies, and solves even pathological examples, i.e., Scarf economies. Our approach extends the use of exploitability-based learning methods from normal-form games to continuous games, exchange markets, and beyond. \nees{} adds to the growing list of differentiable economics methods that aim to provide practitioners with computational tools for the study of economic properties. \nees{} extends
the range of models for which we have  approximate and reliable solvers. 
For instance, existing algorithms only apply to GNE problems that
are  jointly convex, i.e., pseudo-games for which $\actions$ is a convex set, 
and these algorithms are restricted to relaxation methods 
that are not guaranteed to converge in non-jointly convex games. We show that \nees{} can be trained to convergence in pseudo-games under standard assumptions, hence opening the door to solve pseudo-games beyond those that are jointly convex. 

%% file: appendix.tex
\section{Pseudo-Games vs. Games}\label{sec_app:gne_vs_ne}


To see why GNEs cannot be expressed as Nash equilibria, consider a cake cutting problem between two players, in which each agent $\player = 1, 2$ receives payoff $\util[\player](x_\player, x_{-\player}) = x_\player^2 - \frac{1}{1 - x_{-\player}^2}$ where $x_{-\player}$ is the action of $\player$'s opponent. In this problem, each player can request a share of the cake, i.e., for all $\player$, $x_\player \in [0, 1]$, and  the total share of the cake demanded must be less than or equal to 1, i.e., $x_1 + x_2 \leq 1$. This  is a pseudo-game where any player $\player$ can take an action $x_\player \in [0, x_{-\player}]$. A solution can then be modelled as a GNE , i.e., $(x_1^*, x_2^*)$ s.t. $x_\player^* \in \argmax_{x_{\player} \in [0, x_{-\player}^*]} \util[\player](x_{\player}, x_{-\player}^*)$, which corresponds to the set $\simplex[2]$, i.e., the unit simplex in $\R^2$. Although the cake splitting problem cannot be expressed as a game due to the joint constraint $x_1 + x_2 \leq 1$, a common first intuition of many who are familiar with games
is to penalize the payoffs of the players by $-\infty$ for any $(x_1, x_2)$ such that $x_1 + x_2 \geq 1$. This then gives us the following payoff: 
\begin{align}
    \util[\player](x_\player, x_{-\player}) = \left\{ \begin{array}{cc}
        x_\player^2 - \frac{1}{1 - x_{-\player}^2} & \text{ if } x_1 + x_2 \leq 1,\text{ or} \\
        - \infty & \text{otherwise.}
    \end{array} \right.
\end{align}

The Nash equilibria of the game defined by the above payoffs are $\{(x_1, 1) \mid x_1 \in [0,1]\} \cup \{(1, x_2) \mid x_2 \in [0,1]\} \cup \simplex[2]$, and even if the penalty term was more than $-\infty$, the set of Nash equilibria would not be equal to the set of GNE. Note that $\{(x_1, 1) \mid x_1 \in [0,1]\}$ are all Nash equilibria since the payoff of the first player is $-\infty$ not matter what actions it chooses. A similar argument holds for $\{(1, x_2) \mid x_2 \in [0,1]\}$.

\section{GNE Computation Methods Survey}\label{sec_app:related}

Following \citeauthor{arrow-debreu}'s introduction of GNE, \citet{rosen1965gne} initiated the study of the mathematical and computational properties of GNE in pseudo-games with jointly convex constraints, proposing a projected gradient method to compute GNE.
Thirty years later, 
\citet{uryas1994relax} developed the first relaxation methods for finding GNEs, which were improved upon in subsequent works \cite{Krawczyk2000relax, von2009relax}.
Two other types of algorithms were also introduced to the literature: Newton-style methods \cite{facchinei2009generalized, dreves2017computing, von2012newton, izmailov2014error, fischer2016globally, dreves2013newton} and interior-point potential 
methods \cite{dreves2013newton}.
Many of these approaches are based on minimizing the exploitability of the pseudo-game, 
but others use variational inequality \cite{facchinei2007vi, nabetani2011vi} and Lemke methods \cite{Schiro2013lemke}.

Additional, novel methods that transform GNE problems to Nash equilibria problems have
also been analyzed.
These models take the form of either exact penalization methods, which lift the constraints into the objective function via a penalty term 
\cite{facchinei2011partial, fukushima2011restricted, kanzow2018augmented, ba2020exact,  facchinei2010penalty}, or augmented Lagrangian methods \cite{pang2005quasi, kanzow2016multiplier, kanzow2018augmented, bueno2019alm}, which do the same but augmented by dual Lagrangian variables.
Using these methods, \citet{jordan2022gne}  provide the first convergence rates to a $\varepsilon$-GNE in monotone (resp.~strongly monotone) pseudo-games with jointly affine constraints in $\tilde{O}(\nicefrac{1}{\varepsilon})$ ($\tilde{O}(\nicefrac{1}{\sqrt{\varepsilon}})$) iterations.
These algorithms, despite being highly efficient in theory, are numerically unstable in practice \cite{jordan2022gne}.
Nearly all of the aforementioned approaches concern pseudo-games with jointly convex constraints.  \citet{goktas2021minmax} also introduce first-order methods to minimize exploitability in a large class of jointly convex pseudo-games in polynomial-time.

\section{GNE Applications}\label{sec_app:application_motivation}

Some   economic applications of pseudo-games and GNE solvers 
include energy resource allocation \cite{hobbs2007power, wei1999power}, environmental protection \cite{breton2006game, Krawczyk2005environ}, cloud computing \cite{ardagna2017cloud, ardagna2011cloud}, ride sharing services \cite{ban2019rides}, transportation \cite{stein2018transport}, capacity allocation in wireless and network communication \cite{han2011wireless, pang2008distributed}, and  applications to machine learning such as adversarial classification \cite{bruckner2012nash, bruckner2009nash}. Competitive equilibria concepts have also  been used to solve many problems in resource allocation  \cite{varian1973equity, gutman2012fair}, with specific applications to college course allocation \cite{ budish2011combinatorial}, pricing of cloud computing \cite{lai2005tycoon, zahedi2018amdahl},  ad market platforms \cite{conitzer2022pacing},  economic forecasting \cite{partridge2010computable}, and economic policy assessment \cite{dixon1996computable}. 
 
One of the other motivations of research in pseudo-games  is their mathematical significance for general equilibrium theory, the branch of microeconomics which models the interactions between economic agents~\cite{facchinei2010generalized}. General equilibrium theory is a cornerstone of economic theory \cite{debreu1996general}, and is also widely used in policy analysis \cite{dixon1996computable}. The most established general equilibrium model is the Arrow-Debreu model of a competitive economy \cite{arrow-debreu}, which is an instantiation of a pseudo-game in which a seller sets prices for commodities, a set of firms choose what quantity of each commodity to produce, and a set of consumers choose what quantity of each commodity to consume in exchange for their endowment of each commodity. 
This model is a pseudo-game, rather than a game, because the prices set by the sellers determine the value of the consumers' endowments, i.e., their budget, which in turn determines the consumptions of goods they can afford.

The canonical solution concept for this model, and other general equilibrium models more broadly \cite{facchinei2010generalized}, is a \mydef{competitive equilibrium (CE)} \cite{walras,arrow-debreu}.
Here, there is a collection of demands, one per consumer, a collection of supplies, one per firm, and prices, one per commodity, such that given equilibrium prices: 1) no consumer can increase their utility by unilaterally deviating to a consumption  they can afford, 2) no firm can increase their profit by deviating to another feasible production schedule, and 3) the aggregate demand for each commodity (i.e., the sum of the commodity's consumption across all consumers) is equal to its aggregate supply (i.e., the sum of the commodity's production and endowments across firms and consumers respectively), while the total cost of the aggregate demand is equal to the total cost of the aggregate supply 
.
CE are intrinsically related to GNE since the set CE of an Arrow-Debreu competitive economy corresponds to the set of GNE of the corresponding pseudo-game. This approach also works for general equilibrium models more broadly: assuming local non-satiation of consumer preferences, and given a general equilibrium model, one can construct an associated pseudo-game and show that the set of \ce{}  is equal to the set of \gne{} of the associated pseudo-game.

\paragraph{Potential Applications for \nees{}.}
 Applications of economic equilibrium concepts such as GNE and CE often require a decision maker to solve a fixed parametric model either 
1) {\em en masse} over a distribution of parameters 
or 2) quickly in an iterative fashion for a sequence of different parameters. 

The first common use case for the former setting occurs in internet platforms that have to price advertisers in exchange for ad impressions. One standard approach to solve this problem is to let advertisers compete in sequential first price auctions, where winning each auction gives the advertiser the right to show their ad to the website visitor associated with the auction \cite{bigler2019rolling}. As the number of auctions that each advertiser participates in is enormous, the bidding procedure on these platforms is automated. However, beyond certain large advertising companies, it is in general hard for advertisers to come-up with effective automated bidding strategies, as a result companies provide their own bidding strategies to advertisers. One example of these strategies are first-price pacing equilibria,
in which the platform seeks  a vector of pacing multipliers, one for each advertiser, and buyers bid  their value times their pacing multiplier. These pacing multipliers correspond to CE \cite{conitzer2022pacing, conitzer2022multiplicative}, but for large platforms many ad markets run simultaneously, which requires these platforms to solve for CE {\em en masse}.

A second common use case, is in computable general equilibrium, i.e., the study of economic data through the lens of general equilibrium theory, 
which uses CE to make economic forecasts \cite{dixon1996computable}. In these applications, in order to understand the takeaways from a general equilibrium model on the economy, one fixes certain parameters of the model and varies others to understand the consequences of a change in parts of the economy
on the economy as whole. This practice of {\em comparative statics} \cite{nachbar2002general}, requires once again to solve the model for a family of parameters 
{\em en masse}.

For the setting which requires fast and iterative computation of GNE and CE,
a common application is in the context of shared computational resources on platforms such as Mesos \cite{hindman2011mesos}, Quincy \cite{isard2009quincy}, Kubernetes \cite{burns2016borg}, and Yarn \cite{vavilapalli2013apache}. 
To this end, a long line of work has studied resource sharing on computing clusters \cite{chen2018scheduling, ghodsi2011dominant, ghodsi2013choosy, parkes2015beyond}, with many 
methods making use of the repeated computation of competitive equilibrium \cite{gutman2012fair, lai2005tycoon, budish2011combinatorial, zahedi2018amdahl, varian1973equity} or generalized Nash equilibria \cite{ardagna2017cloud, ardagna2011cloud, ardagna2011flexible, anselmi2014generalized}. 
In such settings, as consumers request resources from the platform, the platforms have to compute an equilibrium iteratively while handling all numerical failures within a given time frame. 

Another application that requires one to solve for GNE and CE iteratively is that of state-value function-based reinforcement leaning algorithms for solving market equilibria in stochastic environments, 
such as {\em Model-based Optimistic Online Learning for Markov Exchange Economy} and {\em Model-based Pessimistic Online Learning for Markov Exchange Economy}, which iteratively construct  state-value functions by solving a sequence of competitive equilibrium problems that have
to be solved quickly \cite{liu2022welfare}. Other related
online learning algorithms such as {\em randomized exchange equilibrium learning} \cite{guo2022onlineEE}, which compute a competitive equilibrium when agents' payoffs can  be obtained from sample observations, also require solving for a competitive equilibrium iteratively and quickly.

\section{Non-Lipschitz Exploitability and Unbounded Gradients}\label{sec_app:examples}


\begin{example}\label{ex:non_lipschitz_exploitability_2}
Consider a two-player pseudo-game with action space $\actionspace[1] = \actionspace[2] = 
[0, 1]$, payoffs $\util[1](\action[1][ ], \action[2][ ]) = \action[1][ ]$ $\util[2](\action[1][ ], \action[2][ ]) =\action[2][ ]$, and constraints $\actionconstr[1][ ](\action[1][ ], \action[2][ ]) = \action[2][ ] - 
\action[1][ ]^2 $, $\actionconstr[2][ ](\action[1][ ], \action[2][ ]) = \action[1][ ] - \action[2][ ]^2$. The exploitability of this pseudo-game for all $\action = [0,1]^2$ is given by 
\begin{align}
    \exploit(\action) &=  \max_{(\otheraction[1][ ], \otheraction[2][ ]) : \substack{\action[2][ ] - \otheraction[1][ ]^2 \geq 0\\ \action[1][ ] - \otheraction[2][ ]^2 \geq 0}} \otheraction[1][ ] + \otheraction[2][ ] - \action[1][ ] - \action[2][ ]\\
    &= \max_{(\otheraction[1][ ], \otheraction[2][ ]) : \substack{\sqrt{\action[2][ ]} \geq \otheraction[1][ ] \\  
    =\sqrt{\action[1][ ]} \geq \otheraction[2][ ]}} \otheraction[1][ ] + \otheraction[2][ ] - \action[1][ ] - \action[2][ ]\\
    &= \sqrt{\action[1][ ]} + \sqrt{\action[2][ ]} - \action[1][ ] - \action[2][ ] 
\end{align}
\noindent
The gradient of the exploitability, when it exists is given by $\frac{\partial\exploit}{\partial \action[i][ ]}(\action) = \frac{1}{2\sqrt{\action[i][ ]}} - 1$ for $i = 1, 2$. Note that the payoffs $\util$ and constraints $\actionconstr$ are both Lipschitz-continuous over $[0,1]^2$, however, whenever $\action[1][ ] = 0$ or $\action[2][ ] = 0$, the exploitability grows unboundedly, i.e., if for some $\action[i][ ] \to 0$, then $\frac{\partial\exploit}{\partial \action[i][ ]}(\action) \to \infty$, and hence exploitability is not Lipschitz continuous, and its gradients cannot be bounded over the set $[0, 1]$.
\end{example}

Example~\ref{ex:non_lipschitz_exploitability_2} shows that exploitability in pseudo-games behaves differently than in normal-form games, where exploitability is Lipschitz-continuous. 
The fact that gradients are unbounded means in turn turn that gradients can explode during training if the GNE is located near the non-differentiability (\Cref{ex:non_lipschitz_exploitability_2} shows this can happen, since the GNE strategy for either player can occur at 0). 
As a result, first-order methods can fail, not only theoretically but also in practice.

\section{Additional Preliminary Definitions}\label{sec_ap:prelims}

\paragraph{Notation.}
We use caligraphic uppercase letters to denote sets (e.g., $\calX$); bold lowercase letters to denote vectors (e.g., $\price, \bm \pi$);
bold uppercase letters to denote matrices (e.g., $\allocation$, $\bm \Gamma$),
lowercase letters to denote scalar quantities (e.g., $x, \delta$).
We denote the $i$th row vector of a matrix (e.g., $\allocation$) by the corresponding bold lowercase letter with subscript $i$ (e.g., $\allocation[\buyer])$. 
Similarly, we denote the $j$th entry of a vector (e.g., $\price$ or $\allocation[\buyer]$) by the corresponding Roman lowercase letter with subscript $j$ (e.g., $\price[\good]$ or $\allocation[\buyer][\good]$).

\paragraph{Models.}

An $\varepsilon$-\mydef{variational equilibrium (VE)} (or \mydef{normalized GNE}) of a pseudo-game is a strategy profile $\action^* \in \actions$ s.t.\ for all $\player \in \players$ and $\action \in \actions$, $\util[\player](\action^*) \geq  \util[\player](\action[\player], \naction[\player][][][*]) - \varepsilon$.
We note that in the above definitions, one could just as well write $\action^* \in \actions(\action^*)$ as $\action^* \in \actions$, as any fixed point of the joint action correspondence is also a jointly feasible action profile, and vice versa.
A VE is an $\epsilon$-VE with $\varepsilon = 0$. 
Under our assumptions, while GNE are guaranteed to exist in all pseudo-games by \citeauthor{arrow-debreu}'s lemma on abstract economies \cite{arrow-debreu}, VE are only guaranteed to exist in pseudo-games with jointly convex constraints 
\cite{von2009optimization}.
Note that the set of $\varepsilon$-VE of a pseudo-game is a subset of the set of the $\varepsilon$-GNE, as $\actions(\action^*) \subseteq \actions$, for all $\action^*$ which are GNE of $\pgame$.
The converse, however, is not true, unless $\actionspace \subseteq \actions$.
Further, when $\pgame$ is a game, GNE and VE coincide; we refer to this set simply as NE.

\paragraph{Mathematical Preliminaries.}
Fix any norm $\| \cdot \|$. Given $\calA \subset \R^\outerdim$, the function $\obj: \calA \to \R$ is said to be $\lipschitz[\obj]$-\mydef{Lipschitz-continuous} iff $\forall \outer_1, \outer_2 \in \calX, \left\| \obj(\outer_1) - \obj(\outer_2) \right\| \leq \lipschitz[\obj] \left\| \outer_1 - \outer_2 \right\|$.
Consider a function $\obj: \calX \to \calY$, we denote its Lipschitz-continuity constant by $\lipschitz[\obj]$.
If the gradient of $\obj$ is $\lipschitz[\grad \obj]$-Lipschitz-continuous, we then refer to $\obj$ as $\lipschitz[\grad \obj]$-\mydef{Lipschitz-smooth}. A function $\obj: \calX \to \R$ is said to be \mydef{convex} if $\obj(\x)  \geq \obj(\y) + \grad \obj(\y) \cdot (\x - \y)$, for all $\x, \y \in \calX$ and concave if $-\obj$ is convex. A function $\obj: \calA \to \R$ is said to be $\scparam$-\mydef{Polyak-Lojasiewicz (PL)} if for all $\outer \in \calX$, $\nicefrac{1}{2}\left\| \grad \obj(\outer)\right\|_2^2 \geq \scparam(\obj(\outer) - \min_{\outer \in \calX} \obj(\outer))$.
A function $\obj: \calA \to \R$ is said to be $\scparam$-\mydef{quadratically growing (QG)}, if for all $\outer \in \calX$, $\obj(\outer) - \min_{\outer \in \calX} \obj(\outer) \geq \nicefrac{\scparam}{2}\left\|\outer^* - \outer \right\|^2$ where $\outer^* \in \argmin_{\outer \in \calX} \obj(\outer)$.

\section{Ommited Results and Proofs}\label{sec_app:proofs}

We first revisit the proof of the observation that is central to \nees{}.
\obsganeq*
\begin{proof}[Proof of Observation \ref{obs:gan_eq}]
\begin{align}
    &\min_{\hypothesis \in \actions[][\pgames]} \Ex_{\pgame \sim \distribpgames} \left[\exploit[][\pgame](\hypothesis(\pgame)) \right]\\
    &= \min_{\hypothesis \in \actions[][\pgames]} \Ex_{\pgame \sim \distribpgames} \left[\sum_{\player \in \players} \max_{\otheraction[\player] \in \actions[\player](\naction[\player])} \regret[\player][\pgame](\action[\player], \otheraction[\player]; \naction[\player])\right]\\
    &= \min_{\hypothesis \in \actions[][\pgames]} \Ex_{\pgame \sim \distribpgames} \left[ \max_{\otheraction \in \actions(\action)} \sum_{\player \in \players} \regret[\player][\pgame](\action[\player], \otheraction[\player]; \naction[\player])\right]\\
    &= \min_{\hypothesis \in \actions[][\pgames]} \Ex_{\pgame \sim \distribpgames} \left[\max_{\otheraction \in \actions(\hypothesis(\pgame))} \cumulregret[][\pgame](\hypothesis(\pgame), \otheraction)\right]\\ 
    &= \min_{\hypothesis \in \actions[][\pgames]} \Ex_{\pgame \sim \distribpgames} \left[ \max_{\otherhypothesis[][\pgame] \in \actionspace^\pgames: \otherhypothesis[][\pgame](\pgame) \in \actions[][\pgame](\hypothesis(\pgame))} \cumulregret[][\pgame](\hypothesis(\pgame), \otherhypothesis[][\pgame](\pgame))\right]\\
    &= \min_{\hypothesis \in \actions[][\pgames]} \max_{\otherhypothesis \in \actionspace^{\pgames}: \forall \pgame \in \pgames, \otherhypothesis(\pgame, \hypothesis(\pgame)) \in \actions[][\pgame](\hypothesis(\pgame))}  \Ex_{\pgame \sim \distribpgames} \left[ \cumulregret[][\pgame](\hypothesis(\pgame), \otherhypothesis(\pgame, \hypothesis(\pgame)))\right]
\end{align}
\end{proof}

We now present the proof of \Cref{thm:convergence_stationary_approx}. At a high-level, the proof of the  theorem requires one to first bound the error in the gradient of empirical regret as a function of $\numiters[\otherhypothesis]$, this requires us to derive a gradient domination condition (also known as the PL condition \cite{karimi2016linear}). With such a
lemma in hand, we then can obtain a progress lemma for $\hypothesis$, and prove the theorem.
We first restate the following known lemma, which will become useful in proving convergence of \Cref{alg:cumul_regret}.

\begin{lemma}[Corollary of Theorem 2 \cite{karimi2016linear}]\label{lemma:qg}
If a function $\obj$ is $\scparam$-PL, then $\obj$ is $4\scparam$-quadratically-growing.
\end{lemma}

\begin{restatable}[Inner Loop Error Bound]{lemma}{lemmarrrorbound}
\label{lemma:error_bound}
Suppose that \Cref{assum:convergence} holds. Let $\singularval[\mathrm{min}](\otherhypothesis)$ be the smallest non-zero singular value of $\otherhypothesis$. If \Cref{alg:cumul_regret} is run with learning rates $\forall \iter \in \numiters[\hypothesis], s \in \numiters[\otherhypothesis], \learnrate[\hypothesis][\iter] > 0$ and $\learnrate[\otherhypothesis][s] = \frac{2 s + 1}{2\singularval[\mathrm{min}](\otherhypothesis)\scparam(s + 1)^2}$, for any number of outer loop iterations $\numiters[\hypothesis] \in \N_{++}$, and for $\frac{\lipschitz[{\grad \avg[\cumulregret]}]^{3} \lipschitz[{\avg[\cumulregret]}]^2 }{2 \singularval[\mathrm{min}](\otherhypothesis)\scparam \varepsilon^2}$ total inner loop iterations, where $\varepsilon > 0$.
Then, the outputs $(\weight[][\hypothesis][\iter], \weight[][\otherhypothesis][\iter])_{\iter = 1}^{\numiters[\hypothesis]}$ satisfy
$$
    \left\| \grad \avg[\exploit](\weight[][\hypothesis]) - \grad[{\weight[][\hypothesis]}] \avg[\cumulregret]\left(\weight[][\hypothesis][\iter], \weight[][\otherhypothesis][\iter] \right) \right\|_2 \leq \varepsilon.
$$
\end{restatable}

\begin{proof}[Proof of \Cref{lemma:error_bound}]
Let $\weight[][\otherhypothesis]^*(\weight[][\hypothesis]) \in \argmax_{\weight[][\otherhypothesis] \in \R^\numweights: \forall \pgame \in \pgames, \otherhypothesis(\pgame; \weight[][\otherhypothesis]) \in \actions(\hypothesis(\pgame, \weight[][\hypothesis]))} \avg[\cumulregret]\left(\weight[][\hypothesis], \weight[][\otherhypothesis]) \right)$. 
Let $\singularval[\mathrm{min}](\otherhypothesis)$ be the smallest non-zero singular value of $\otherhypothesis$. For all $\weight[][\hypothesis] \in \R^\numweights$, $\avg[\cumulregret]\left(\weight[][\hypothesis], \cdot \right)$ is the composition of $\Ex\left[\cumulregret(\action, \cdot)\right]$ for all $\action \in \actions$, a $\scparam$-strongly concave function (\Cref{assum:convergence}), with an affine function, $\otherhypothesis(\pgame, \cdot)$, which means that for all $\weight[][\hypothesis] \in \R^\numweights$, $\avg[\cumulregret]\left(\weight[][\hypothesis], \cdot \right)$ is a $\singularval[\mathrm{min}](\otherhypothesis)\scparam$-PL function (see Appendix B, \cite{karimi2016linear}), and the following convergence bound holds if $\learnrate[\otherhypothesis][\iter] = \frac{2 \iter + 1}{2\singularval[\mathrm{min}](\otherhypothesis)\scparam(\iter + 1)^2}$ for the inner loop iterates as a corollary of of convergence results on stochastic gradient ascent for PL objectives (Theorem 4, \cite{karimi2016linear}):
\begin{align}
    \avg[\exploit](\weight[][\hypothesis][\iter]) - \avg[\cumulregret]\left(\weight[][\hypothesis][\iter], \weight[][\otherhypothesis][\iter] \right) \leq 
    \frac{\lipschitz[{\grad\avg[\cumulregret] }]  \lipschitz[{\avg[\cumulregret]}]^2 }{2 \singularval[\mathrm{min}](\otherhypothesis)\scparam \numiters[\otherhypothesis]}
\end{align}
Since for all $\weight[][\hypothesis] \in \R^\numweights$, 
is $\singularval[\mathrm{min}](\otherhypothesis)\scparam$-PL, by \Cref{lemma:qg}, we have:
\begin{align}
    \avg[\exploit](\weight[][\hypothesis][\iter]) - \avg[\cumulregret]\left(\weight[][\hypothesis][\iter], \weight[][\otherhypothesis][\iter] \right)  \geq 4\singularval[\mathrm{min}](\otherhypothesis)\scparam\left\|\weight[][\otherhypothesis]^*(\weight[][\hypothesis][\iter]) - \weight[][\otherhypothesis][\iter] \right\|_2^2 \enspace ,
\end{align}
Combining the two previous inequalities, we get:
\begin{align}
    4\singularval[\mathrm{min}](\otherhypothesis)\scparam\left\|\weight[][\otherhypothesis]^*(\weight[][\hypothesis][\iter]) - \weight[][\otherhypothesis][\iter] \right\|_2^2 \leq  
    \frac{\lipschitz[{\grad\avg[\cumulregret]}]  \lipschitz[{\avg[\cumulregret]}]^2 }{2 \singularval[\mathrm{min}](\otherhypothesis)\scparam \numiters[\otherhypothesis]}\\
    \left\|\weight[][\otherhypothesis]^*(\weight[][\hypothesis][\iter]) - \weight[][\otherhypothesis][\iter] \right\|_2^2 \leq  
    \frac{\lipschitz[{\grad\avg[\cumulregret]}]  \lipschitz[{\avg[\cumulregret]}]^2 }{8 \singularval[\mathrm{min}](\otherhypothesis)^2\scparam^2 \numiters[\otherhypothesis]}\\
    \left\|\weight[][\otherhypothesis]^*(\weight[][\hypothesis][\iter]) - \weight[][\otherhypothesis][\iter] \right\|_2 \leq  
    \frac{ \lipschitz[{\avg[\cumulregret] }]}{2 \singularval[\mathrm{min}](\otherhypothesis)\scparam }\sqrt{\frac{\lipschitz[{\grad\avg[\cumulregret]}]}{2\numiters[\otherhypothesis]}} \label{eq:bound_approx_sol}
\end{align}

Finally, we bound the error between the approximate gradient computed by \Cref{alg:cumul_regret} and the true gradient $\grad \avg[\exploit](\weight[][\hypothesis])$ at each iteration $\iter \in \N_{++}$.  Note that $\grad \avg[\cumulregret]$ is Lipschitz-continuous in $(\weight[][\hypothesis], \weight[][\otherhypothesis])$ since the composition of Lipschitz continuous functions is Lipschitz. Hence, $\grad \avg[\cumulregret]$ is also Lipschitz and we have:
\begin{align}
    &=\left\| \grad[{\weight[][\hypothesis]}] \avg[\exploit]\left(\weight[][\hypothesis][\iter] \right) - \grad[{\weight[][\hypothesis]}] \avg[\cumulregret]\left(\weight[][\hypothesis][\iter], \weight[][\otherhypothesis][\iter] \right) \right\|_2\\
    &\leq \left\| \grad[{(\weight[][\hypothesis], \weight[][\otherhypothesis])}] \avg[\cumulregret]\left(\weight[][\hypothesis][\iter], \weight[][\otherhypothesis]^*(\weight[][\hypothesis]) \right) - \grad[{(\weight[][\hypothesis], \weight[][\otherhypothesis])}] \avg[\cumulregret]\left(\weight[][\hypothesis][\iter], \weight[][\otherhypothesis][\iter] \right) \right\|_2\\
    &\leq \lipschitz[{\grad \avg[\cumulregret]}] \left\| (\weight[][\hypothesis][\iter], \weight[][\otherhypothesis]^*(\weight[][\hypothesis][\iter])) -  (\weight[][\hypothesis][\iter], \weight[][\otherhypothesis][\iter]) \right\|_2\\
    &\leq \lipschitz[{\grad \avg[\cumulregret]}] \left(\left\| \weight[][\hypothesis][\iter] - \weight[][\hypothesis][\iter]\right\|_2 + \left\| \weight[][\otherhypothesis]^*(\weight[][\hypothesis][\iter]) -  \weight[][\otherhypothesis][\iter] \right\|_2 \right)\\
    &= \lipschitz[{\grad \avg[\cumulregret]}] \left\| \weight[][\otherhypothesis]^*(\weight[][\hypothesis][\iter]) -  \weight[][\otherhypothesis][\iter] \right\|_2 \\
    &\leq \frac{ \lipschitz[{\grad \avg[\cumulregret]}]^{\nicefrac{3}{2}} \lipschitz[{\avg[\cumulregret]}]}{2 \singularval[\mathrm{min}](\otherhypothesis)\scparam \sqrt{2\numiters[\otherhypothesis]}} && \text{(\Cref{eq:bound_approx_sol})}
\end{align}

Then, given $\varepsilon > 0$, for any number of inner loop iterations such that $\numiters[\otherhypothesis] \geq \frac{\lipschitz[{\grad \avg[\cumulregret]}]^{3} \lipschitz[{\avg[\cumulregret]}]^2 }{2 \singularval[\mathrm{min}](\otherhypothesis)\scparam \varepsilon^2}$, for all $\iter \in [\numiters[\hypothesis]]$, we have:

\begin{align}
     \left\| \grad \avg[\exploit](\weight[][\hypothesis]) - \grad[{\weight[][\hypothesis]}] \avg[\cumulregret]\left(\weight[][\hypothesis][\iter], \weight[][\otherhypothesis][\iter] \right) \right\|_2 \leq \varepsilon
\end{align}
\end{proof}

Using the above gradient error bound we derive a progress bound for the outer loop of our algorithm.

\begin{lemma}[Progress Lemma for Approximate Iterate]\label{lemma:progress_bound_approximate}
Suppose that  \Cref{assum:convergence} holds. Let $\singularval[\mathrm{min}](\otherhypothesis)$ be the smallest non-zero singular value of $\otherhypothesis$. If \Cref{alg:cumul_regret} is run with learning rates $\forall \iter \in [\numiters[\hypothesis]], s \in [\numiters[\otherhypothesis]],  \learnrate[\hypothesis][\iter] > 0$ and $\learnrate[\otherhypothesis][s] = \frac{2 s + 1}{2\singularval[\mathrm{min}](\otherhypothesis)\scparam(s + 1)^2}$, for any number of outer loop iterations $\numiters[\hypothesis] \in \N_{++}$, and for $\numiters[\otherhypothesis] \geq \frac{\lipschitz[{\grad \avg[\cumulregret]}]^{3} \lipschitz[{\avg[\cumulregret]}]^2 }{2 \singularval[\mathrm{min}](\otherhypothesis)\scparam \varepsilon}$ total inner loop iterations, where $\varepsilon > 0$.
Then, the outputs $(\weight[][\hypothesis][\iter], \weight[][\otherhypothesis][\iter])_{\iter = 1}^{\numiters[\hypothesis]}$ satisfy:
\begin{align}
    \avg[\exploit](\weight[][\hypothesis][\iter + 1])
    &\leq \avg[\exploit](\weight[][\hypothesis][\iter]) - \learnrate[\hypothesis ][\iter] \left\| \grad[\hypothesis] \avg[\exploit](\weight[][\hypothesis][\iter]) \right\|_2^2 + \learnrate[\hypothesis ][\iter]\varepsilon +  \frac{(\learnrate[\hypothesis][\iter])^2 \lipschitz[{\avg[\cumulregret]}]  \lipschitz[{\grad \avg[\exploit]}]^2}{2} 
\end{align}
\end{lemma}

\begin{proof}[Proof of \Cref{lemma:progress_bound_approximate}]
Fix $\iter \in \numiters[\hypothesis]$. 
Define $\err[\iter] \doteq \grad \avg[\exploit](\weight[][\hypothesis][\iter]) - \grad[{\weight[][\hypothesis]}] \avg[\cumulregret](\weight[][\hypothesis][\iter], \weight[][\otherhypothesis][\iter])$.
Since $\avg[\exploit]$ is also Lipschitz-smooth, we have that:
\begin{align}
    &\avg[\exploit](\weight[][\hypothesis][\iter + 1])\\ 
    &\leq \avg[\exploit](\weight[][\hypothesis][\iter]) + \left<\grad[\hypothesis] \avg[\exploit](\weight[][\hypothesis][\iter]), \weight[][\hypothesis][\iter + 1] - \weight[][\hypothesis][\iter] \right> +  \nicefrac{\lipschitz[{\grad \avg[\exploit]}]}{2} \left\|\weight[][\hypothesis][\iter + 1] - \weight[][\hypothesis][\iter] \right\|_2^2\\
    &\leq \avg[\exploit](\weight[][\hypothesis][\iter]) + \left<  \grad\avg[\exploit](\weight[][\hypothesis][\iter]), - \learnrate[\hypothesis][\iter]  \grad[{\weight[][\hypothesis]}] \nicefrac{1}{\batchpgames[\hypothesis][\iter]} \sum_{\pgame \in  \batchpgames[\hypothesis][\iter] } \cumulregret[\pgame](\weight[][\hypothesis][\iter], \weight[][\otherhypothesis][\iter]) \right> + \nicefrac{\lipschitz[{\grad \avg[\exploit]}]}{2} \left\|\learnrate[\hypothesis ][\iter]  \grad[{\weight[][\hypothesis]}] \nicefrac{1}{\batchpgames[\hypothesis][\iter]} \sum_{\pgame \in  \batchpgames[\hypothesis][\iter] } \cumulregret[\pgame](\weight[][\hypothesis][\iter], \weight[][\otherhypothesis][\iter]) \right\|_2^2\\
    &\leq \avg[\exploit](\weight[][\hypothesis][\iter]) - \learnrate[\hypothesis][\iter]  \left<  \grad\avg[\exploit](\weight[][\hypothesis][\iter]), \grad[{\weight[][\hypothesis]}] \nicefrac{1}{\batchpgames[\hypothesis][\iter]} \sum_{\pgame \in  \batchpgames[\hypothesis][\iter] } \cumulregret[\pgame](\weight[][\hypothesis][\iter], \weight[][\otherhypothesis][\iter]) \right>  + \frac{(\learnrate[\hypothesis][\iter])^2 \lipschitz[{\avg[\cumulregret]}]  \lipschitz[{\grad \avg[\exploit]}]^2}{2}
\end{align}
where the last line follows from $\cumulregret[\pgame]$ being $\lipschitz[{\avg[\cumulregret]}]$-Lipschitz continuous. Taking the expectation w.r.t $\batchpgames[\hypothesis][\iter]$ conditioned on $(\weight[][\hypothesis][\iter], \weight[][\otherhypothesis][\iter])$, we get:
\begin{align}
    &\leq \avg[\exploit](\weight[][\hypothesis][\iter]) - \learnrate[\hypothesis ][\iter] \left\|  \grad[{\weight[][\hypothesis]}] \avg[\cumulregret](\weight[][\hypothesis][\iter], \weight[][\otherhypothesis][\iter]) \right\|_2^2 +    \frac{(\learnrate[\hypothesis][\iter])^2 \lipschitz[{\avg[\cumulregret]}]  \lipschitz[{\grad \avg[\exploit]}]^2}{2}\\ &= \avg[\exploit](\weight[][\hypothesis][\iter]) - \learnrate[\hypothesis ][\iter] \left\| \grad[\hypothesis] \avg[\exploit](\weight[][\hypothesis][\iter]) - \grad[\hypothesis] \avg[\exploit](\weight[][\hypothesis][\iter]) +  \grad[{\weight[][\hypothesis]}] \avg[\cumulregret](\weight[][\hypothesis][\iter], \weight[][\otherhypothesis][\iter]) \right\|_2^2 +   \frac{(\learnrate[\hypothesis][\iter])^2 \lipschitz[{\avg[\cumulregret]}]  \lipschitz[{\grad \avg[\exploit]}]^2}{2}\\
    &\leq \avg[\exploit](\weight[][\hypothesis][\iter]) - \learnrate[\hypothesis ][\iter] \left\| \grad[\hypothesis] \avg[\exploit](\weight[][\hypothesis][\iter]) \right\|^2_2 + \learnrate[\hypothesis ][\iter]\left\| \grad[\hypothesis] \avg[\exploit](\weight[][\hypothesis][\iter]) -  \grad[{\weight[][\hypothesis]}] \avg[\cumulregret](\weight[][\hypothesis][\iter], \weight[][\otherhypothesis][\iter]) \right\|_2^2 + \frac{(\learnrate[\hypothesis][\iter])^2 \lipschitz[{\avg[\cumulregret]}]  \lipschitz[{\grad \avg[\exploit]}]^2}{2}\\
    &\leq \avg[\exploit](\weight[][\hypothesis][\iter]) - \learnrate[\hypothesis ][\iter] \left\| \grad[\action] \avg[\exploit](\weight[][\hypothesis][\iter]) \right\|_2^2 + \learnrate[\hypothesis ][\iter](\err[\iter])^2 +    \frac{(\learnrate[\hypothesis][\iter])^2 \lipschitz[{\avg[\cumulregret]}]  \lipschitz[{\grad \avg[\exploit]}]^2}{2}
\end{align}

By \Cref{lemma:error_bound}, we then have $(\err[\iter])^2 \leq \varepsilon$, which gives us for all $\iter \in \numiters[\hypothesis]$:
\begin{align}
    &\leq \avg[\exploit](\weight[][\hypothesis][\iter]) - \learnrate[\hypothesis ][\iter] \left\| \grad[\action] \avg[\exploit](\weight[][\hypothesis][\iter]) \right\|_2^2 + \learnrate[\hypothesis ][\iter]\varepsilon^2 +    \frac{(\learnrate[\hypothesis][\iter])^2 \lipschitz[{\avg[\cumulregret]}]  \lipschitz[{\grad \avg[\exploit]}]^2}{2}
\end{align}

Substituting $\delta = \varepsilon^2$, we obtain the lemma's statement.

\end{proof}

Finally, telescoping the the inequality given in the above lemma we can obtain our convergence result.

\thmconvergencestationaryapprox*
\begin{proof}[Proof of \Cref{thm:convergence_stationary_approx}]
By \Cref{lemma:progress_bound_approximate}, we have:
\begin{align}
    \avg[\exploit](\weight[][\hypothesis][\iter + 1]) 
   &\leq \avg[\exploit](\weight[][\hypothesis][\iter]) - \learnrate[\hypothesis ][\iter] \left\| \grad[\action] \avg[\exploit](\weight[][\hypothesis][\iter]) \right\|_2^2 +   \learnrate[\hypothesis ][\iter] \varepsilon + \frac{(\learnrate[\hypothesis][\iter])^2 \lipschitz[{\avg[\cumulregret]}]  \lipschitz[{\grad \avg[\exploit]}]^2}{2} 
\end{align}

Summing up the inequalities for $\iter = 0, \hdots, \numiters[\hypothesis] - 1$:
\begin{align}
    \sum_{\iter = 1}^{\numiters[\hypothesis]} \learnrate[\hypothesis ][\iter] \left\| \grad[\action] \avg[\exploit](\weight[][\hypothesis][\iter]) \right\|_2^2 &\leq \avg[\exploit](\weight[][\hypothesis][0]) -  \avg[\exploit](\weight[][\hypothesis][{\numiters[\hypothesis]}])  + \sum_{\iter = 1}^{\numiters[\hypothesis]} \learnrate[\hypothesis ][\iter] \varepsilon + \sum_{\iter = 1}^{\numiters[\hypothesis]} (\learnrate[\hypothesis ][\iter])^2   \frac{\lipschitz[{\avg[\cumulregret]}]  \lipschitz[{\grad \avg[\exploit]}]^2}{2}
\end{align}
Taking the minimum of $ \left\| \grad[\action] \avg[\exploit](\weight[][\hypothesis][\iter]) \right\|_2^2$ across all $\iter \in [\numiters[\hypothesis]]$, to obtain:

\begin{align}
    \left( \min_{\iter = 0, \hdots, \numiters[\hypothesis] - 1} \left\| \grad[\action] \avg[\exploit](\weight[][\hypothesis][\iter]) \right\|_2^2 \right)  \sum_{\iter = 1}^{\numiters[\hypothesis]} \learnrate[\hypothesis ][\iter] \leq  \avg[\exploit](\weight[][\hypothesis][0]) -  \avg[\exploit](\weight[][\hypothesis][{\numiters[\hypothesis]}])   + \sum_{\iter = 1}^{\numiters[\hypothesis]} \learnrate[\hypothesis ][\iter]   \varepsilon + \sum_{\iter = 1}^{\numiters[\hypothesis]} (\learnrate[\hypothesis ][\iter])^2   \frac{\lipschitz[{\avg[\cumulregret]}]  \lipschitz[{\grad \avg[\exploit]}]^2}{2}\\
    \left( \min_{\iter = 0, \hdots, \numiters[\hypothesis] - 1} \left\| \grad[\action] \avg[\exploit](\weight[][\hypothesis][\iter]) \right\|_2^2 \right)   \leq \frac{\avg[\exploit](\weight[][\hypothesis][0]) -  \avg[\exploit](\weight[][\hypothesis][{\numiters[\hypothesis]}])   + \sum_{\iter = 1}^{\numiters[\hypothesis]} \learnrate[\hypothesis ][\iter]    \varepsilon + \sum_{\iter = 1}^{\numiters[\hypothesis]} (\learnrate[\hypothesis ][\iter])^2   \frac{\lipschitz[{\avg[\cumulregret]}]  \lipschitz[{\grad \avg[\exploit]}]^2}{2}}{\sum_{\iter = 1}^{\numiters[\hypothesis]} \learnrate[\hypothesis ][\iter]}\\
    \left( \min_{\iter = 0, \hdots, \numiters[\hypothesis] - 1} \left\| \grad[\action] \avg[\exploit](\weight[][\hypothesis][\iter]) \right\|_2^2 \right)   \leq \frac{ \avg[\exploit](\weight[][\hypothesis][0]) -  \avg[\exploit](\weight[][\hypothesis][{\numiters[\hypothesis]}])}{\sum_{\iter = 1}^{\numiters[\hypothesis]} \learnrate[\hypothesis ][\iter]}   + \frac{\lipschitz[{\avg[\cumulregret]}]  \lipschitz[{\grad \avg[\exploit]}]^2}{2} \frac{\sum_{\iter = 1}^{\numiters[\hypothesis]} (\learnrate[\hypothesis ][\iter])^2 }{\sum_{\iter = 1}^{\numiters[\hypothesis]} \learnrate[\hypothesis ][\iter]} + \varepsilon
\end{align}

Suppose that $\learnrate[\hypothesis ][\iter] = \frac{1}{\sqrt{\iter}}$, we then have:
\begin{align}
    \left( \min_{\iter = 0, \hdots, \numiters[\hypothesis] - 1} \left\| \grad[\action] \avg[\exploit](\weight[][\hypothesis][\iter]) \right\|_2^2 \right) \in O\left(\frac{\log(\numiters[\hypothesis])}{\sqrt{\numiters[\hypothesis]}} + \varepsilon\right)
\end{align}
\end{proof}

This result characterizes the generalization capacity of the generator and discriminator as a function of the covering number of the hypothesis classes. The proof relies on the following lemma, which states that the distance of any generator and discriminator from the hypothesis class to their $\coveringnumber$-cover is bounded in payoff space.

\begin{restatable}[Bounded Expected Cumulative Regret]{lemma}{lemmaboundedcumulregret}\label{lemma:bounded_cumul_regret}
Suppose that \Cref{assum:sample} holds. For any hypothesis classes $\hypotheses \subset \actions[][\pgames]$, $\otherhypotheses \subset \actionspace^{\actions \times \pgames}$ and hypotheses $\hypothesis \in \hypotheses$ and $\otherhypothesis \in \otherhypotheses$, any pseudo-game distribution $\distribpgames \in \simplex(\pgames)$: 
\begin{align}
    \left| \Ex_{\pgame \sim \distribpgames} \left[\cumulregret[][\pgame](\hypothesis^\coveringnumber(\pgame), \otherhypothesis^\coveringnumber(\pgame, \hypothesis^\coveringnumber(\pgame))) \right] -  \Ex_{\pgame \sim \distribpgames} \left[\cumulregret[][\pgame](\hypothesis(\pgame), \otherhypothesis(\pgame, \hypothesis(\pgame)) \right]\right| \leq 2 \lipschitz[{\cumulregret[]}]  \coveringnumber
\end{align}
\end{restatable}

\begin{proof}[Proof of \Cref{lemma:bounded_cumul_regret}]
Let $\lipschitz[{\cumulregret[][\pgame]}] = \max_{(\action, \otheraction) \in \actionspace \times \actionspace} \left| \grad \cumulregret[][\pgame](\action, \otheraction) \right| = \max_{\pgame \in \pgames} \max_{(\action, \otheraction) \in \actionspace \times \actionspace}  \left|\grad  \util[\player][][\pgame](\otheraction[\player], \naction[\player]) -  \grad \util[\player][][\pgame](\action) \right|$ and let \\ $\lipschitz[{\cumulregret[]}] = \max_{\pgame \in \pgames} \max_{(\action, \otheraction) \in \actionspace \times \actionspace} \left| \grad \cumulregret[][\pgame](\action, \otheraction) \right| = \max_{\pgame \in \pgames} \max_{(\action, \otheraction) \in \actionspace \times \actionspace} \left|  \grad \util[\player][][\pgame](\otheraction[\player], \naction[\player]) -  \grad \util[\player][][\pgame](\action) \right|$.
\begin{align}
    &\left| \Ex_{\pgame \sim \distribpgames}     \left[\cumulregret[][\pgame](\hypothesis^\coveringnumber(\pgame), \otherhypothesis^\coveringnumber(\pgame, \hypothesis^\coveringnumber(\pgame))) \right] -  \Ex_{\pgame \sim \distribpgames} \left[\cumulregret[][\pgame](\hypothesis(\pgame), \otherhypothesis(\pgame, \hypothesis(\pgame)) \right]\right|\\ &=  \left| \Ex_{\pgame \sim \distribpgames}     \left[\cumulregret[][\pgame](\hypothesis^\coveringnumber(\pgame), \otherhypothesis^\coveringnumber(\pgame, \hypothesis^\coveringnumber(\pgame))) - \cumulregret[][\pgame](\hypothesis(\pgame), \otherhypothesis(\pgame, \hypothesis(\pgame))) \right]\right|\\
    &\leq \max_{\pgame \in \pgames} \left| \cumulregret[][\pgame](\hypothesis^\coveringnumber(\pgame), \otherhypothesis^\coveringnumber(\pgame, \hypothesis^\coveringnumber(\pgame))) - \cumulregret[][\pgame](\hypothesis(\pgame), \otherhypothesis(\pgame, \hypothesis(\pgame))) \right|\\
    &\leq \max_{\pgame \in \pgames} \lipschitz[{\cumulregret[][\pgame]}]\left\| (\hypothesis^\coveringnumber(\pgame), \otherhypothesis^\coveringnumber(\pgame, \hypothesis^\coveringnumber(\pgame))) - (\hypothesis(\pgame), \otherhypothesis(\pgame, \hypothesis(\pgame))) \right\|\\
    &= \lipschitz[{\cumulregret[]}]\left\| (\hypothesis^\coveringnumber(\pgame), \otherhypothesis^\coveringnumber(\pgame, \hypothesis^\coveringnumber(\pgame))) - (\hypothesis(\pgame), \otherhypothesis(\pgame, \hypothesis(\pgame))) \right\|\\
    &= \lipschitz[{\cumulregret[]}]  \left( \left\| \hypothesis^\coveringnumber(\pgame) - \hypothesis(\pgame) \right\| + \left\| \otherhypothesis^\coveringnumber(\pgame, \hypothesis^\coveringnumber(\pgame)) -\otherhypothesis(\pgame, \hypothesis(\pgame))) \right\| \right)\\
    &\leq 2 \lipschitz[{\cumulregret[]}]  \coveringnumber
\end{align}
\end{proof}

\thmsamplecomplexity*

\begin{proof}[Proof of \Cref{thm:sample_complexity}]
Let $\lipschitz[{\cumulregret[]}] = \max\limits_{\pgame \in \pgames} \max\limits_{(\action, \otheraction) \in \actionspace \times \actionspace} \left| \grad \cumulregret[][\pgame](\action, \otheraction) \right| \leq \max\limits_{\pgame \in \pgames} \max\limits_{(\action, \otheraction) \in \actionspace \times \actionspace} \grad \left|  \util[\player][][\pgame](\otheraction[\player], \naction[\player]) -  \util[\player][][\pgame](\action) \right|$.
Consider $\coveringnumber = \frac{\varepsilon}{6 \lipschitz[{\cumulregret[]}]}$, let $\hypotheses^\coveringnumber$ and $\otherhypotheses^\coveringnumber$ be the minimum $\coveringnumber$-covering sets of $\hypotheses$ and $\otherhypotheses$ respectively, and let $\hypothesis^\coveringnumber$ and $\otherhypothesis^\coveringnumber$ be the closest hypotheses to $\hypothesis$ and $\otherhypothesis$  in the sets $\hypotheses^\coveringnumber$ and $\otherhypotheses^\coveringnumber$ respectively, i.e., $\hypothesis^\coveringnumber \in \argmin_{\hypothesis^\prime \in \hypotheses^\coveringnumber} \left\| \hypothesis - \hypothesis^\prime \right\|$ and $\otherhypothesis^\coveringnumber \in \argmin_{\otherhypothesis^\prime \in \otherhypotheses^\coveringnumber} \left\| \otherhypothesis - \otherhypothesis^\prime \right\|$.
\begin{align}
    &\Pr_{\samplepgames \sim \distribpgames^\numsamples} \left[ \exists (\hypothesis, \otherhypothesis) \in \hypotheses \times \otherhypotheses, \left| \Ex_{\pgame \sim \mathrm{unif}(\samplepgames)} \left[\cumulregret[][\pgame](\hypothesis(\pgame), \otherhypothesis(\pgame, \hypothesis(\pgame))) \right] -  \Ex_{\pgame \sim \distribpgames} \left[\cumulregret[][\pgame](\hypothesis(\pgame), \otherhypothesis(\pgame, \hypothesis(\pgame))) \right] \right| \geq \varepsilon  \right]\\
    &\leq  \Pr_{\samplepgames \sim \distribpgames^\numsamples} \left[  \exists (\hypothesis, \otherhypothesis) \in \hypotheses \times \otherhypotheses, \left| \Ex_{\pgame \sim \mathrm{unif}(\samplepgames)} \left[\cumulregret[][\pgame](\hypothesis(\pgame), \otherhypothesis(\pgame, \hypothesis(\pgame))) \right] -  \Ex_{\pgame \sim \mathrm{unif}(\samplepgames)} \left[\cumulregret[][\pgame](\hypothesis^\coveringnumber(\pgame), \otherhypothesis^\coveringnumber(\pgame, \hypothesis^\coveringnumber(\pgame)) \right] \right| \right. \notag \\
    & \quad + \left. \left| \Ex_{\pgame \sim \mathrm{unif}(\samplepgames)} \left[\cumulregret[][\pgame](\hypothesis^\coveringnumber(\pgame), \otherhypothesis^\coveringnumber(\pgame, \hypothesis^\coveringnumber(\pgame))) \right] -  \Ex_{\pgame \sim \distribpgames} \left[\cumulregret[][\pgame](\hypothesis^\coveringnumber(\pgame), \otherhypothesis^\coveringnumber(\pgame, \hypothesis^\coveringnumber(\pgame))) \right] \right| + \right. \notag \\ 
    & \quad \left. \left| \Ex_{\pgame \sim \distribpgames} \left[\cumulregret[][\pgame](\hypothesis^\coveringnumber(\pgame), \otherhypothesis^\coveringnumber(\pgame, \hypothesis^\coveringnumber(\pgame))) \right] -  \Ex_{\pgame \sim \distribpgames} \left[\cumulregret[][\pgame](\hypothesis(\pgame), \otherhypothesis(\pgame, \hypothesis(\pgame)) \right]\right|   \geq \varepsilon  \right]\\
    &\leq  \Pr_{\samplepgames \sim \distribpgames^\numsamples} \left[  \exists (\hypothesis, \otherhypothesis) \in \hypotheses \times \otherhypotheses, 2\lipschitz[{\cumulregret[][]}] \coveringnumber + \left| \Ex_{\pgame \sim \mathrm{unif}(\samplepgames)} \left[\cumulregret[][\pgame](\hypothesis^\coveringnumber(\pgame), \otherhypothesis^\coveringnumber(\pgame, \hypothesis^\coveringnumber(\pgame))) \right] -  \Ex_{\pgame \sim \distribpgames} \left[\cumulregret[][\pgame](\hypothesis^\coveringnumber(\pgame),  \otherhypothesis^\coveringnumber(\pgame, \hypothesis^\coveringnumber(\pgame))) \right] \right| + 2\lipschitz[{\cumulregret[][]}] \coveringnumber \geq \varepsilon  \right]\\
    &= \Pr_{\samplepgames \sim \distribpgames^\numsamples} \left[  \exists (\hypothesis^\coveringnumber, \otherhypothesis^\coveringnumber) \in \hypotheses^\coveringnumber \times \otherhypotheses^\coveringnumber, \left| \Ex_{\pgame \sim \mathrm{unif}(\samplepgames)} \left[\cumulregret[][\pgame](\hypothesis^\coveringnumber(\pgame), \otherhypothesis^\coveringnumber(\pgame, \hypothesis^\coveringnumber(\pgame))) \right] -  \Ex_{\pgame \sim \distribpgames} \left[\cumulregret[][\pgame](\hypothesis^\coveringnumber(\pgame), \otherhypothesis^\coveringnumber(\pgame, \hypothesis^\coveringnumber(\pgame))) \right] \right| \geq \frac{\varepsilon}{3}  \right]
\end{align}
where the penultimate line follow from the Lipschitz continuity of the Nash approximation error, and the final line from $\coveringnumber = \frac{\varepsilon}{6}$.

Then, using a union bound we get:
\begin{align}
    & \Pr_{\samplepgames \sim \distribpgames^\numsamples} \left[ \forall (\hypothesis^\coveringnumber, \otherhypothesis^\coveringnumber) \in \hypotheses^\coveringnumber \times \otherhypotheses^\coveringnumber,, \left| \Ex_{\pgame \sim \mathrm{unif}(\samplepgames)} \left[\cumulregret[][\pgame](\hypothesis^\coveringnumber(\pgame), \otherhypothesis^\coveringnumber(\pgame, \hypothesis^\coveringnumber(\pgame))) \right] -  \Ex_{\pgame \sim \distribpgames} \left[\cumulregret[][\pgame](\hypothesis^\coveringnumber(\pgame), \otherhypothesis^\coveringnumber(\pgame, \hypothesis^\coveringnumber(\pgame))) \right] \right| \leq  \frac{\varepsilon}{3}  \right]\\
    &\leq \coveringcard(\hypotheses, \nicefrac{\varepsilon}{6}) \coveringcard(\otherhypotheses, \nicefrac{\varepsilon}{6}) \Pr_{\samplepgames \sim \distribpgames^\numsamples} \left[ \left| \Ex_{\pgame \sim \mathrm{unif}(\samplepgames)} \left[\cumulregret[][\pgame](\hypothesis^\coveringnumber(\pgame), \otherhypothesis^\coveringnumber(\pgame, \hypothesis^\coveringnumber(\pgame))) \right] -  \Ex_{\pgame \sim \distribpgames} \left[\cumulregret[][\pgame](\hypothesis^\coveringnumber(\pgame), \otherhypothesis^\coveringnumber(\pgame, \hypothesis^\coveringnumber(\pgame))) \right] \right| \leq \frac{\varepsilon}{3}  \right]\\
    &\leq 2 \coveringcard(\hypotheses, \nicefrac{\varepsilon}{6}) \coveringcard(\otherhypotheses, \nicefrac{\varepsilon}{6}) \exp\left\{ - 2\numsamples\left(\frac{\varepsilon}{3}\right)^2 \right\}
\end{align}
Re-organizing expressions, in order to get that: $$\Pr_{\samplepgames \sim \distribpgames^\numsamples} \left[ \exists (\hypothesis, \otherhypothesis) \in \hypotheses \times \hypotheses, \left| \Ex_{\pgame \sim \mathrm{unif}(\samplepgames)} \left[\cumulregret[][\pgame](\hypothesis(\pgame), \otherhypothesis(\pgame, \hypothesis(\pgame))) \right] -  \Ex_{\pgame \sim \distribpgames} \left[\cumulregret[][\pgame](\hypothesis(\pgame), \otherhypothesis(\pgame, \hypothesis(\pgame))) \right] \right| \geq \varepsilon  \right] = \delta$$ we obtain that the sample size $\numsamples$ should be set as follows:
\begin{align}
    \numsamples \geq \frac{9 }{\varepsilon^2} \left[ \log\left(\frac{\coveringcard(\hypotheses, \nicefrac{\varepsilon}{6})}{\delta}\right) + \log\left(\frac{\coveringcard(\otherhypotheses, \nicefrac{\varepsilon}{6})}{\delta}\right) \right]
\end{align}
\end{proof}

\section{Experiments}\label{sec_app:experiments}

We run three sets of experiments
in which we train \nees{} in three different pseudo-game settings.
All  experiments are run with 5 randomly selected different seeds ($\{5, 10, 25, 30, 42\}$), with hyperparameter selection being done over 
all 5 seeds. Unless otherwise mentioned, all results correspond to an average over these 5 seeds, with confidence intervals reported across these seeds as appropriate. 
In all of our experiments, we adopt the update rule in ADAM for the gradient step, making
use of the ADAM implementation in the OPTAX library. 
We use JaxOPT's projected gradient method to compute best-responses and thus the  exploitability 
of an action profile
when a closed form is not available for the best-response.
For all of the networks used in our experiments, if BatchNorm is used, it is applied before the activation layer. We describe whether if BatchNorm is used in the architecture of each network individually in the following sections. 

\paragraph{Computational Resources.}
Our normal-form game experiments were run on GPUs
while our other experiments were run on CPUs. 

\paragraph{Programming Languages, Packages, and Licensing.}
We ran our experiments in Python 3.7 \cite{van1995python}, using NumPy \cite{numpy}, Jax \cite{jax2018github}, OPTAX \cite{jax2018github}, Haiku \cite{haiku2020github},  JaxOPT \cite{jaxopt_implicit_diff}, and pycdd \cite{pycddlib2015}.
All figures were graphed using Matplotlib \cite{matplotlib} and Seaborn \cite{Waskom2021seaborn}. 

Numpy and Seaborn are distributed under a liberal BSD license. Matplotlib only uses BSD compatible code, and its license is based on the PSF license. CVXPY is licensed under an APACHE license. Jax and Haiku are licensed under the Apache 2.0 License. Pycdd is distributed under a GNU license.

\input{gamut}

\newpage

\input{arrow}

\newpage
\input{kyoto}

%% file: gamut.tex
\subsection{Normal-form Games} 

Our first set of experiments aims to explore the impact of the accuracy of the discriminator on the performance of the generator. To this end, we consider normal-form games in which there exists a closed form solution for the discriminator. We observe that with an accurate enough discriminator, \nees{} achieves a performance similar to the neural architecture proposed by  \citet{duan2021pac} when using the same equilibrium mapping architecture for the generator.  

In our experiments, we replicate the setup proposed by \citet{duan2021pac}, and we try to solve five
games from the GAMUT library, namely, Traveller's Dilemma, Bertrand Oligopoly, Grab the Dollar, War of Attrition, and Majority Voting. We give a description of each game, as presented by \citet{duan2021pac}:
\begin{itemize}
    \item \mydef{Traveler's dilemma}: Each player simultaneously requests an amount of money and receives the lowest of the requests submitted by all players.
    \item \mydef{Grab the dollar}: A price is up for grabs, and both players have to decide when to grab the price. The action of each player is the chosen times. If both players grab it simultaneously, they will rip the price and receive a low payoff. If one chooses a time earlier than the other, they will receive the high payoff, and the opposing player will receive a payoff between the high and the low.
    \item \mydef{War of attrition}: In this game, both players compete for a single object, and each chooses a time to concede the object to the other player. If both concede at the same time, they share the object. Each player has a valuation of the object, and each player’s utility is decremented at every time step.
    \item \mydef{Bertrand oligopoly}: All players in this game are producing the same item and are expected to set a price at which to sell the item. The player with the lowest price gets all the demand for the item and produces enough items to meet the demand to obtain the corresponding payoff.
    \item \mydef{Majority Voting}: This is an $n$-player symmetric game. All players vote for one of the other players. Players’ utilities for each candidate being declared the winner are arbitrary. If there is a tie, the winner is the candidate with the lowest number. There may be multiple Nash equilibria in this game.
\end{itemize}

\paragraph{Game Generation.}

We use the {\em GAMUT library} \cite{nudelman2004run}, which is a normal-form game generation library designated for testing game-theoretic algorithms, to generate a data set of games. 
Following \citet{duan2021pac}, the games were generated so that payoffs were normalized between $[0,1]$, with all other parameters drawn randomly. We generate 2000 games with 2 players and 300 actions for both players and for each game category, setting aside 200 for validation and 100 for testing. 

\paragraph{Network Architecture.}

In all our experiments, we use a generator for \nees{} that has the same architecture as the equilibrium solver proposed by \citet{duan2021pac}. Namely, we use a neural network with 4 hidden layers each with 1024 nodes and ReLU activations. The final layers of the generator have the same dimension as an action profile and each action in the profile is passed through a player-wise softmax activation.
We augment the entire network with BatchNorm layers with non-trainable parameters after each activation layer. The total number of parameters for this generator is 20,855.

To explore how the accuracy of \nees{}'s generator degrades as one uses more and more approximate discriminators, we consider 4 types of  discriminators: a true discriminator that takes as input every player's expected payoffs for all actions and outputs the action with the highest payoff (this discriminator recovers exactly the \citet{duan2021pac} network architecture);
a softmax discriminator with a scaling parameter of $100$, i.e., $\mathrm{softmax}(\u) = \frac{e^{100u_i}}{\sum_i e^{100u_i}}$, that takes as input the expected payoffs of each player at the given equilibrium actions predicted by the generator and outputs its softmax;
a less precise softmax discriminator with a scaling parameter of $10$;
and,  a neural network with one linear layer with 1024 nodes and a softmax activation layer with scaling 1 (the total number of parameters for this discriminator is 7,115).

\paragraph{Training and Hyperparameters.}

We run our algorithm with no inner loop iterations and 10,000 outer loop iterations for the non-neural network discriminators, since they require no training. We adopt for the gradient step in our algorithm 
the ADAM algorithm. We use a learning rate of 0.001 for the optimizer step on each of 
the generator and discriminator  (when appropriate), and use the default settings for the other hyperparameters of ADAM as given in the OPTAX implementation. We process the training data in batches of size 50. 

\paragraph{Experimental Results.} 

In \Cref{fig:gamut_train} (train) and \Cref{fig:gamut_test} (test), we observe that as the quality of the discriminator becomes more approximate the quality of the generator degrades significantly in certain games, underscoring the importance of designing optimal discriminators for \nees{} in certain settings.
Perhaps more interestingly, we observe that  \nees{} with a linear layer discriminator has a hard time in the Bertrand oligopoly and Traveller's Dillemma games. The reason for this is that the discriminator has a hard time learning a pure strategy best-response action, and this seems  crucial for our training algorithm to find an optimal generator in these two games.
This is further justified by \nees{}'s near perfect performance when coupled with 
the true best-response discriminator. In contrast, for other games the neural network's performance suggests that an approximate best-response is enough for our training algorithm to find the optimal discriminator. Future work can investigate the relationship between game classes, and the precision-level w.r.t.~the discriminator that is required for our training algorithms to perform well. We note that \mydef{normalized exploitability} is given as the exploitability divided by the average exploitability over the action space.
\begin{figure*}[t]
  \centering
  \begin{subfigure}{\textwidth}
    \includegraphics[width=\linewidth,trim={0 0 8.5cm 1cm},clip]{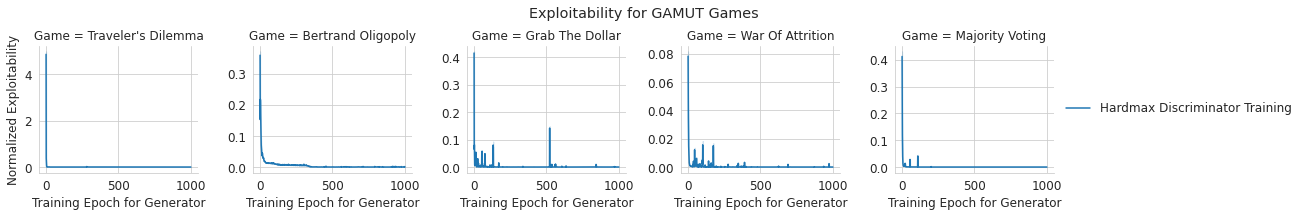}
    \caption{\nees{} with discriminator that outputs the true best-response ($\argmax$ of the expected payoffs for each player).}
  \end{subfigure}
  \hspace{0.5cm}
  \begin{subfigure}{\textwidth}     
    \includegraphics[width=\linewidth,trim={0 0 6.2cm 1cm},clip]{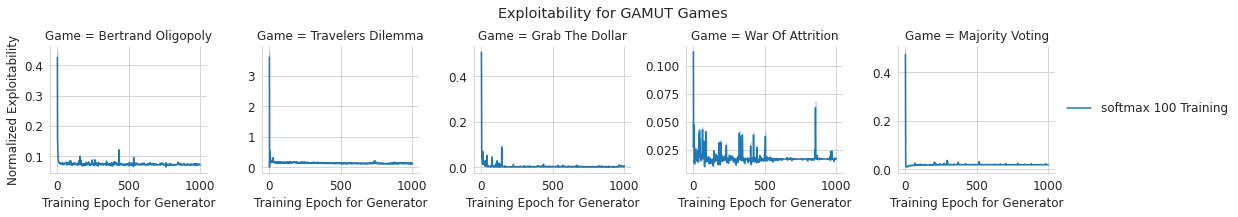}
    \caption{\nees{} with discriminator that outputs the softmax of the expected payoffs for each player (scaling parameter 100).}
  \end{subfigure}
  \begin{subfigure}{\textwidth}
    \includegraphics[width=\linewidth,trim={0 0 8.8cm 1cm},clip]{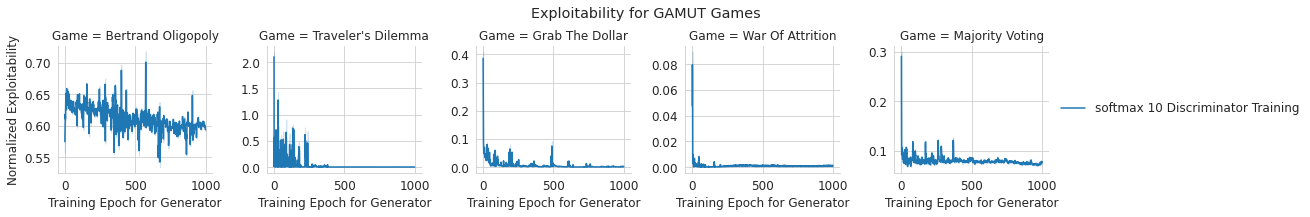}
    \caption{\nees{} with discriminator that outputs the softmax of the expected payoffs for each player (scaling parameter 10).}
  \end{subfigure}

  \begin{subfigure}{\textwidth}
    \includegraphics[width=\linewidth,trim={0 0 11.2cm 1cm},clip]{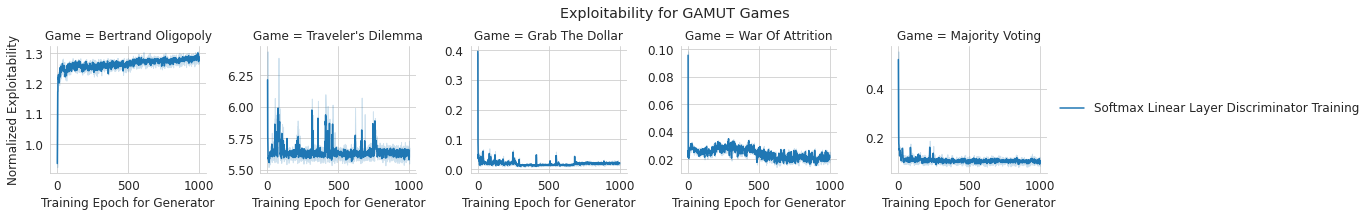}
    \caption{\nees{} with a neural network discriminator.}
  \end{subfigure}
  \caption{Training Exploitability of \nees{} on five classes of GAMUT games. We observe that the performance of \nees{} degrades significantly in Bertrand oligopoly and traveller's dilemma when the discriminator is not precise enough.\label{fig:gamut_train}}
\end{figure*}

\begin{figure*}[!ht]
  \centering
  \begin{subfigure}{\textwidth}
    \includegraphics[width=\linewidth,trim={0 0 8.5cm 1cm},clip]{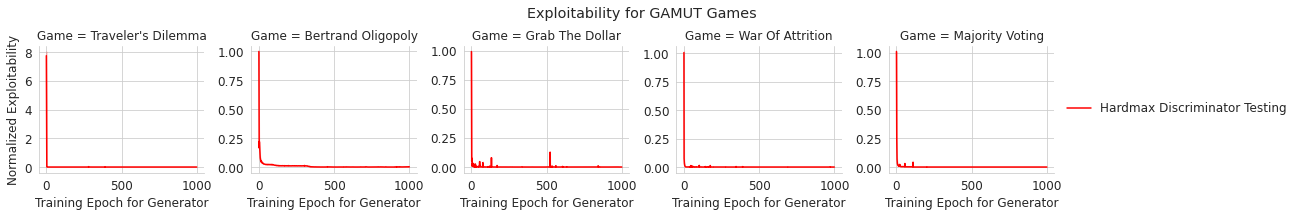}
    \caption{\nees{} with discriminator that outputs the true best-response ($\argmax$ of the expected payoffs for each player).}
  \end{subfigure}
  \hspace{0.5cm}
  \begin{subfigure}{\textwidth}     
\includegraphics[width=\linewidth,trim={0 0 6cm 1cm},clip]{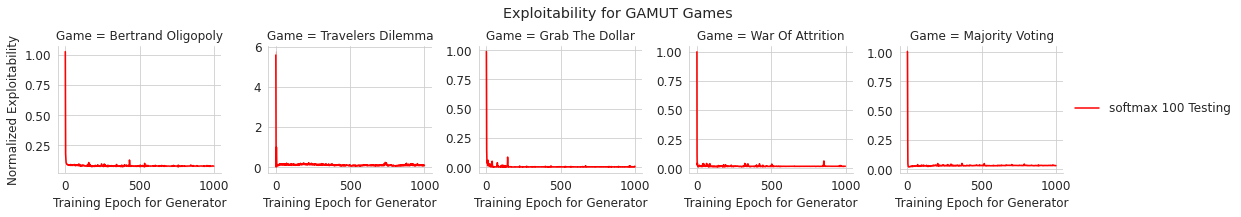}
    \caption{\nees{} with discriminator that outputs the softmax of the expected payoffs for each player (scaling parameter 100).}
  \end{subfigure}
  \begin{subfigure}{\textwidth}
    \includegraphics[width=\linewidth,trim={0 0 8.6cm 1cm},clip]{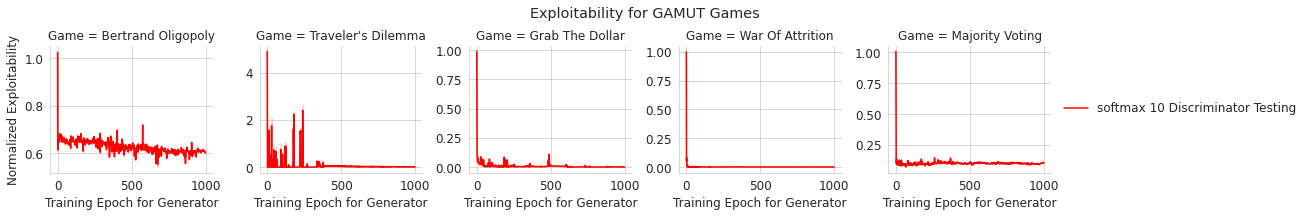}
    \caption{\nees{} with discriminator that outputs the softmax of the expected payoffs for each player (scaling parameter 10).}
  \end{subfigure}

  \begin{subfigure}{\textwidth}
    \includegraphics[width=\linewidth,trim={0 0 11cm 1cm},clip]{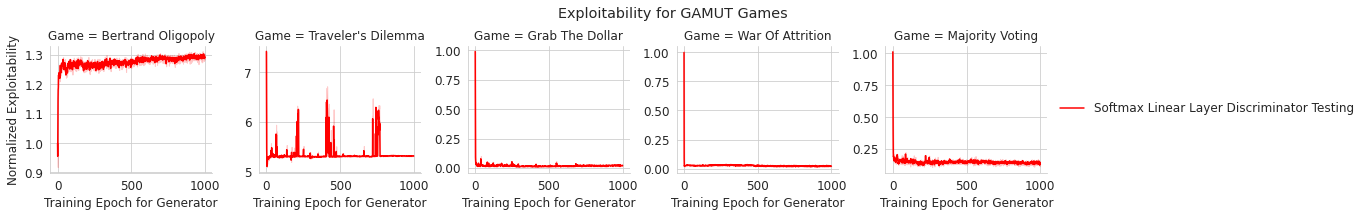}
    \caption{\nees{} with a neural network discriminator.}
  \end{subfigure}
  \caption{Testing Exploitability of \nees{} on five classes of GAMUT games. We observe that the performance of \nees{} degrades significantly in Bertrand oligopoly and traveller's dilemma
  when the discriminator is not precise enough.\label{fig:gamut_test}}
\end{figure*}

%% file: arrow.tex
\subsection{Arrow-Debreu Exchange Economies}\label{sec_ap:arrow_debreu}

\paragraph{Additional Preliminaries.} A \mydef{Competitive equilibrium (CE)} is a tuple which consists of allocations $\allocation \in \R_+^{\numgoods \times \numbuyers}$, and prices $\price \in \R_+^\numgoods$ such that 1.~ all traders $\buyer \in \buyers$ maximize utility constrained by their budget: $\allocation[\buyer]^* \in \argmax_{\allocation[\buyer] \in \consumptions[\buyer]: \allocation[\buyer] \cdot \price \leq \consendow[\buyer] \cdot \price} \util[\buyer](\allocation[\buyer])$, 2.~the markets clear and goods that are not demanded are priced at 0, i.e., $ \sum_{\buyer = 1}^{\numagents} \allocation[\buyer]^*  \leq \sum_{\buyer = 1}^{\numagents} \consendow[\buyer]$ and $\price^* \cdot \left( \sum_{\buyer = 1}^{\numagents} \allocation[\buyer]^*  - \sum_{\buyer = 1}^{\numagents} \consendow[\buyer] \right) = 0$.

A \mydef{Scarf economy}, denoted $(\consendow, \valuation)$, is a Leontief exchange economy with 3 buyers and 3 goods, where the valuation $\valuation$ and endowment $\consendow$ matrices are given as follows:
\begin{align}
    &\consendow = \begin{pmatrix}
        1 & 0 & 0 \\
        0 & 1 & 0 \\
        0 & 0 & 1 \\
    \end{pmatrix} 
    &\valuation = \begin{pmatrix}
        0 & 1 & 0 \\
        0 & 0 & 1 \\
        1 & 0 & 0 \\
    \end{pmatrix}  
\end{align}

\paragraph{Related Work.}

Exchange economies can be solved in polynomial-time via t\^atonnement for CES utilties with $\rho \in [0, 1)$ \cite{bei2015tatonnement}. However,  t\^atonnement is not guaranteed to converge beyond these domains. There exists a convex programs to compute CE in linear exchange markets in polynomial time \cite{devanur2016rational}. The computation of CE is PPAD-complete for Leontief \cite{codenotti2006leontief}, piecewise-linear concave, and additively seperable concave \cite{chen2009settling}, and exchange markets \cite{vazirani2011market, chen2009spending}.
The complexity of CES markets for $\rho \in (\infty, 0)$ is unknown
and remains an open question.

\paragraph{Experimental Setup.}

We consider experiments on six different exchange economies, each 
with 3 buyers and 5 goods,
and with each economy defined by the utility functions of the players: 1) linear, 2) Cobb-Douglas, 3) Leontief, 4) gross substitutes CES where for all buyers $
\buyer \in \buyers$, $\rho_\buyer \in [0.5, 1]$, 5) gross complements CES where for all buyers $ \buyer \in \buyers$, $\rho_\buyer \in [-1.25, -0.75]$, and 6) mixed CES markets where for all buyers $ \buyer \in \buyers$, $\rho_\buyer \in [-1.25, -0.75] \cap  [0.5, 1]$.\footnote{We reduce the range of the $\rho$ parameter to avoid numerical instability in the computation of utilities.} In addition to these settings, we also consider a Leontief economy with 3 buyers and 3 goods with the goal of solving Scarf economies.

\paragraph{Baselines.} 

We benchmark our algorithm to the most well-known algorithm to solve Arrow-Debreu markets, \mydef{t\^atonnement}, which is an auction-like algorithm that is guaranteed to converge for $\rho \in (1, 0]$ \cite{bei2015tatonnement}, and thus including the Cobb-Douglas cases, and 
the \mydef{exploitability descent} algorithm\footnote{This algorithm simply corresponds to running gradient descent on the exploitability. More information on this algorithm can be found in \citeauthor{goktas2022exploitability} \cite{goktas2022exploitability}.}. %
T\^atonnement is defined by the following sequence of prices:
\begin{align}
    \price[][\iter+1] = \price[][\iter] + \eta \sqrt{t} \left( \sum_{\buyer \in \buyers} \allocation[\buyer][][\iter] - \sum_{\buyer \in \buyers} \consendow[\buyer]^{(\iter)} \right) && t = 0, 1, \hdots\\
    \allocation[\buyer][][\iter] \in \argmax_{\allocation[\buyer]: \allocation[\buyer] \cdot \price[][\iter] \leq \consendow[\buyer] \cdot \price[][\iter]} \util[\buyer](\allocation[\buyer]) && \buyer \in \buyers, t = 0, 1, \hdots, 
\end{align}
\noindent while exploitability descent is defined by the following sequence of prices:
\begin{align}
    (\price[][\iter+1], \allocation[][][\iter+1])  = (\price[][\iter], \allocation[][][\iter]) - \eta \sqrt{t} \left(\grad \exploit(\price[][\iter], \allocation[][][\iter]) \right) && t = 0, 1, \hdots 
\end{align}
\noindent where $\exploit(\price[][\iter], \allocation[][][\iter])$ is the exploitability associated with the exchange economy pseudo-game.
For both of these baselines, we use a decreasing learning rate of $\eta \sqrt{t}$ as a function of the number of iterations $t$,
and run the baselines until convergences is observed (we observe that experiments all converge in less than $\leq 200$ iterations). 
We note that the use of a learning rate schedule is necessary 
for these algorithms to converge \cite{goktas2022exploitability}.
We run an extensive grid search for $\eta$ over the set \{1.0, 0.5, 0.1, 0.05, 0.01, 0.005, 0.001, 0.0005, 0.0001]\}, selecting $\eta$ to minimize exploitability over the validation set.
We then evaluate these baselines on the test set  with the selected hyperparameters.

\paragraph{Economy Generation.}

For all experiments, we generate 5,000 exchange economy instances,
and set aside 500 markets for each of validation and test.
To generate a market instance, for all buyers $\buyer \in \buyers$ and goods $\good \in \goods$, we sample  each endowment $\consendow[\buyer][\good] \sim \mathrm{Unif}[10^{-9}, 1]$, valuations $\valuation[\buyer][\good] \sim \mathrm{Unif}[10^{-9}, 1]$, and when appropriate
we sample the substitution parameters $\bm{\rho}$ from the  ranges mentioned above. 
A competitive equilibrium is guaranteed to exist 
for exchange markets in all the exchange markets we sample
by Arrow-Debreu's first existence theorem \cite{arrow-debreu} since buyers are endowed with a non-zero amount of each good.

\paragraph{Generator Architecture.}

\begin{figure}[!ht]
    \centering
    \includegraphics[width=\linewidth]{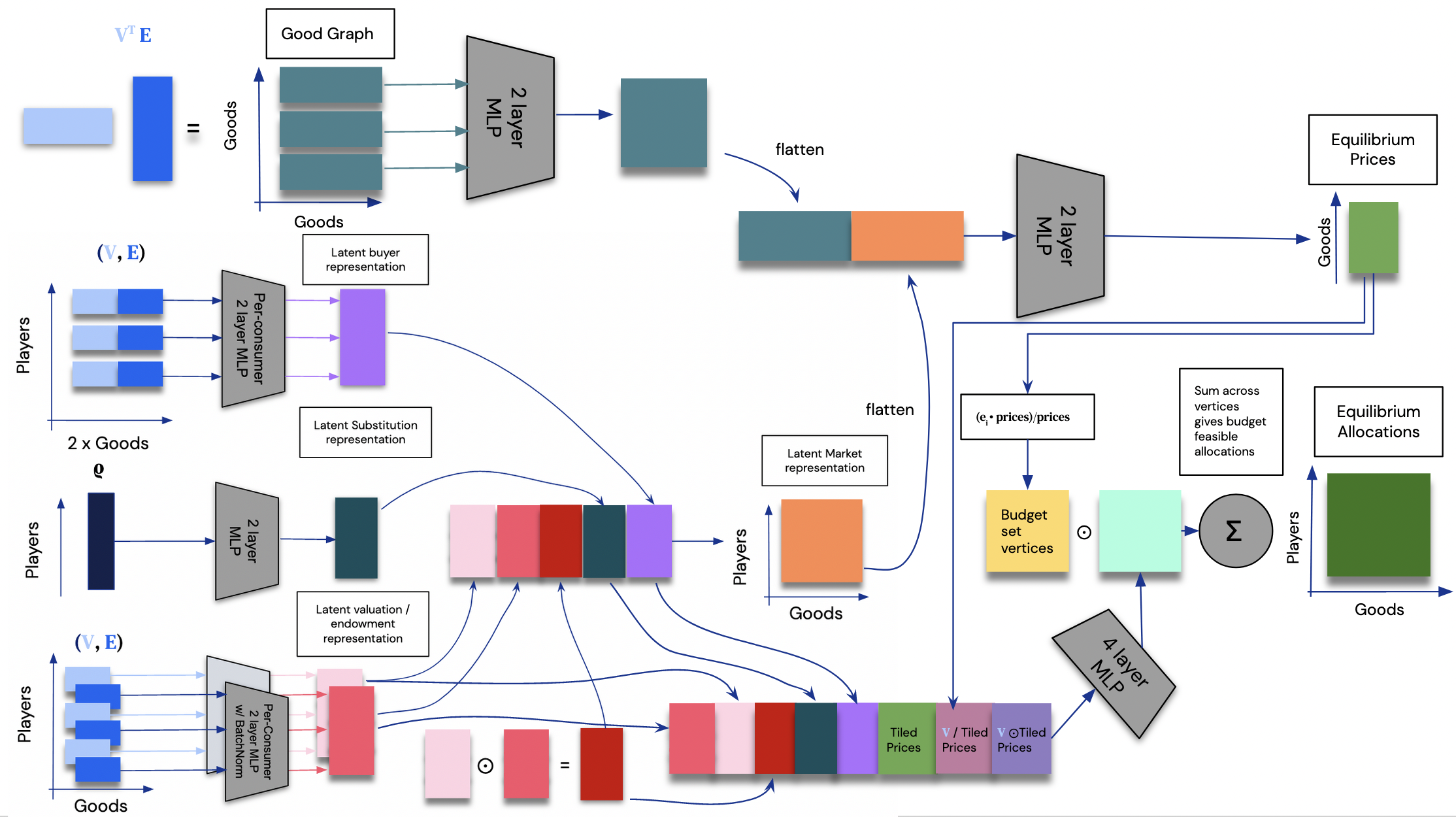}
    \caption{Architecture of the generator of \nees{} for exchange economies.}
    \label{fig:ad_gen}
\end{figure}

We summarize our generator's architecture in \Cref{fig:ad_gen}. As a reminder, in this setting, the generator takes as input an exchange economy $(\consendow, \valuation, \bm{\rho})$\footnote{If the economy is not a CES economy $\bm \rho$ is drawn uniformly at random from $[0.25, 0.75]^\numbuyers$.} and outputs equilibrium prices $\hat{\price}$ and allocations $\hat{\allocation}$. We use the same generator architecture in each of our experiments. 
This generator takes as input the economy matrix $(\consendow, \valuation)$, and passes it through two fully connected layers with ReLU activations and 20 and 10 nodes respectively, to obtain a
\mydef{latent buyer representation}.
The valuations $\valuation$, and endowments $\consendow$ are also separately fed through a network with the same architecture. The network for valuations are augmented with a BatchNorm layer with trainable parameters. This gives us a \mydef{latent valuation representation} and a \mydef{latent  endowment representation}. Each latent representation as well as the element-wise product of  the latent valuations
and latent endowments 
are then concatenated and fed through two fully connected layers with 20 and $\numgoods$ (number of goods)
hidden nodes, respectively, followed by ReLU
activations at each layer. This gives us a 
\mydef{latent market representation}. We then pass the matrix $\valuation^T\consendow$, which we call the {\em good graph}, through two fully connected layers, each
with BatchNorm with trainable parameters and ReLU activation.
These layers have 20 and 10 nodes respectively. We refer to the output of this network as the 
{\em latent good graph}.
We then concatenate the flattened,  latent good graph and latent market representations and feed them through three fully connected layers with 40, 20, and $\numgoods$ (number of goods) nodes, respectively. The outputs of the first two layers are passed through ReLU activations, while the last layer is passed through a softmax. The output of this final layer is the generator's predicted prices $\hat{\price}$. 

Given these prices, we then build an {\em allocation coefficient matrix} of dimensions $\numbuyers \times \numgoods$, where the $(i,j)th$ entry is given by $\frac{\consendow[\buyer] \cdot \hat{\price}}{\hat{\price[\good]}}$. We then calculate the budgets $\consendow \hat{\price}$ of the buyers at prices $\hat{\price}$, and feed them through two fully connected layers, with 30 and 20 nodes respectively and ReLU activations, to obtain a {\em latent budget representation}.
Define the {\em tiled prices} as $(\hat{\price}, \hdots, \hat{\price})^T \in \R^{\numbuyers \times \numgoods}$, i.e., a matrix whose rows consists of the vector $\hat{\price}$ repeated $\numbuyers$ times so as to obtain ``$\numbuyers$ tiles of prices''.
We concatenate the latent representation of buyers, the tiled prices, the latent endowment representation, the valuations element-wise divided by the tiled prices, the valuations element-wise multiplied by the tiled prices, and the latent budget on the last dimension, i.e., we append each matrix horizontally so as to preserve the number of rows $\numbuyers$. 
We pass the obtained matrix through 3 fully connected layers, with 100, 50, and 20 nodes respectively and ReLU 
activations. Finally, the output of this network is passed through a  fully connected layer with $\numgoods$
nodes and a softmax activation, leaving us with a matrix of dimension $\numplayers \times \numgoods$ whose rows sum up to 1. 
We multiply this matrix element-wise with the allocation coefficient matrix,
which gives us the allocation $\hat{\allocation}$ for the generator. Notice that $\hat{\allocation}$ is budget feasible at price $\hat{\price}$. The total number of parameters of this generator is 20,855.

\paragraph{Discriminator Architectures.}

\begin{figure}[!ht]
    \centering
    \includegraphics[width=\linewidth]{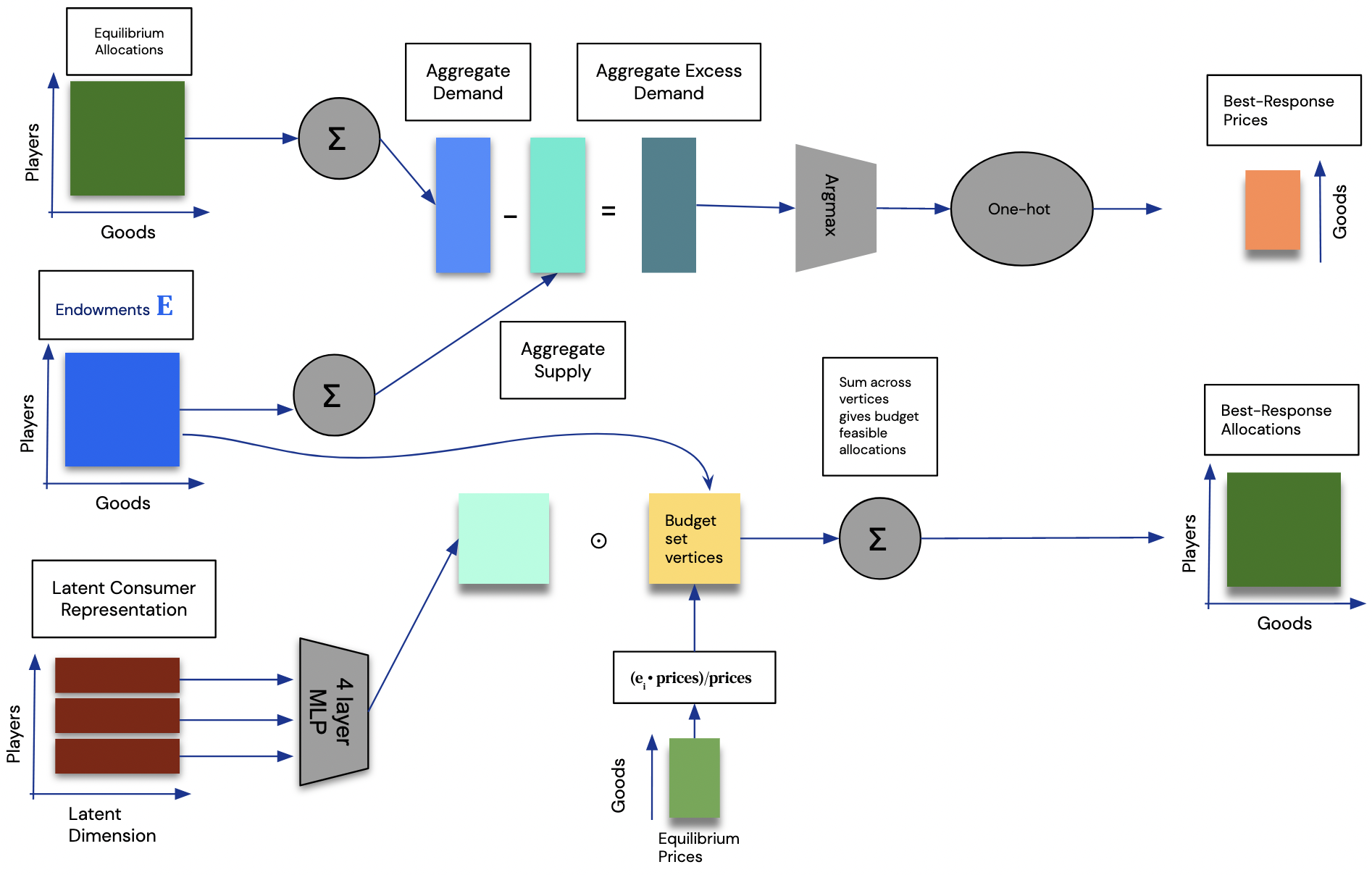}
    \caption{Architecture of the generator of \nees{} for exchange economies. The latent consumer representation associated with each exchange economy is described in the description of the discriminator.}
    \label{fig:ad_discrim}
\end{figure}

We summarize the architecture of our network in \Cref{fig:ad_discrim}. As a reminder, in this setting, the generator takes as input an exchange economy $(\consendow, \valuation, \bm{\rho})$ as well as an equilibrium $(\hat{\price}, \hat{\allocation})$ and outputs best-responses $(\price^*, \allocation^*)$.
We build different, modular discriminator architectures for
each of the linear, Cobb-Douglas, Leontief, and CES markets. 
These networks take as input the market matrix $(\consendow, \valuation, \bm{\rho})$ and the equilibrium $(\hat{\allocation}, \hat{\price})$ predicted by the generator.
For all four discriminators, the discriminator outputs price $\price^*$ such that $\price[\good]^* = 1$ if $\good \in \argmax_{\good \in \goods} \sum_{\buyer \in \buyers}\left( \hat{\allocation[ ][ ]}_{\buyer\good} - \consendow[\buyer][\good]\right)$ and $\price[\good]^* = 0$ otherwise. In regard to the allocations, we build a modular allocation network, which takes as input a latent representation of each consumer
as a matrix of dimension $\numbuyers \times p$, and outputs
an allocation. It is in this latent representation
that the discriminators for each economy differ. We describe the latent representation associated with different exchange economies below. 
We first build an \mydef{allocation coefficient matrix}, where the $(i,j)th$ entry is given by $\frac{\consendow[\buyer] \cdot \hat{\price}}{\hat{\price[\good]}}$. We then pass the matrix of latent consumer representations through three fully connected layers with  100, 50, 20 nodes respectively and ReLU activations.   We take this output and pass it through a final fully connected layer with $\numgoods$ (num goods)
nodes and softmax activations to obtain a matrix of dimension $\numplayers \times \numgoods$ whose rows sum up to 1. We multiply this matrix element-wise with the allocation coefficient matrix,
which gives us {\em allocations} $\allocation^*$ that
are budget feasible at prices $\hat{\price}$.

With the main architecture out of the way, we can now present the 
different latent consumer representations that we
use. For the linear and Cobb-Douglas markets,
the latent representation is simply the matrix that
is given by the valuations matrix $\valuation$ whose rows are divided by $\hat{\price}$. 
For Leontief, the latent representation is the matrix whose $(i,j)th$ entries are given as $\frac{\valuation[\buyer][\good]}{\sum_{\good \in \goods} \price[\good] \valuation[\buyer][\good]}$. Finally, for CES, the latent representation is given as follows: First we obtain latent representations of
each of $\rho$, $\hat{\price}$, $(\consendow[\buyer] \bigodot \price)_{\buyer \in \buyers}^T$, and $\valuation$ by passing them through separate but identical, two fully connected BatchNorm layers with ReLU activations, and
20 and 10 nodes respectively. 
The concatenation of all of these latent representations on the last dimension so as to obtain a matrix with $\numbuyers$ rows and 40 columns, where each row is the concatenation of the latent representations for that consumer gives us the \mydef{latent consumer representation}.

The number of parameters for the discriminator for each of the linear and Cobb-Douglas markets is $6,775$, for Leontief is $7,115$,
and for CES is $14,125$. 

\paragraph{Training Hyperparameters.}

We run our algorithm with an initial warm-up of 10,000 iterations for the discriminator. This warm-up procedure follows exactly the inner loop of \Cref{alg:cumul_regret} but instead uses randomly sampled economies, and randomly sampled action profiles. After the warmup, we use only one step of inner loop iteration for running \Cref{alg:cumul_regret} and run the outer loop for 10,000 iterations. Together with the warmup, the small number of inner loop iterations allow us to significantly speed up the training process. The gradient step is provided by  the ADAM algorithm. For the discriminator, we use the same learning rate for the warm-up and regular training. The learning rates used for ADAM in different markets can be found in \Cref{table:ee_learning_rates}. For all other hyperparameters, we use the default settings of ADAM as given in the OPTAX implementation. 
We process the training data in batches of size 200. We found the learning rates for our algorithm by performing grid search on the validation set for all economies. For the Scarf economy the grid search values were sampled from a standard lognormal distribution.
\begin{table}[t]
\centering
\begin{tabular}{c|c|c|}
 Economy Type &  Generator &  Discriminator  \\\hline
 Linear & 0.0001 & 0.001   \\\hline
 Cobb-Douglas & 0.0001 & 0.00001\\\hline
 Leontief & 0.00001 &  0.01  \\\hline
 GS CES &  0.00001 &  0.0001  \\\hline
 GC CES & 0.0001 &  0.0001 \\\hline
 Mixed CES & 0.0001 &  0.00001 \\\hline
 Scarf & 0.000003297599624930109 & 0.0000014820835507051155
\end{tabular}
\caption{Learning rates used for ADAM to train \nees{} in different markets. These learning rates are found via grid search on the validation set for all economies.\label{table:ee_learning_rates}}
\end{table}

We present additional results, missing from the main body of the paper, 
in \Cref{fig:arrow_exploit} and 
\Cref{fig:violin_ad_ap}. We see that \nees{} outperforms all baselines in all economies except Gross Substitutes CES economies for which t\^atonnement performs best. This makes sense since for Gross Substitutes CES economies t\^atonnement is guaranteed to converge to a competitive equilibrium. Even the, we note that the performance of \nees{} is remarkable since it achieves a testing normalized exploitability of 0.005, meaning that \nees{} finds an action profile which is closer than 99.95\% of allocations and prices to a competitive equilibrium.

\begin{figure*}[!ht]
  \centering
  \begin{subfigure}{0.45\textwidth}
    \includegraphics[width=\linewidth]{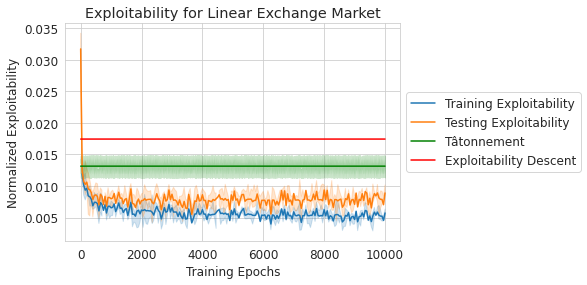}
  \end{subfigure}
  \begin{subfigure}{0.45\textwidth}
    \includegraphics[width=\linewidth]{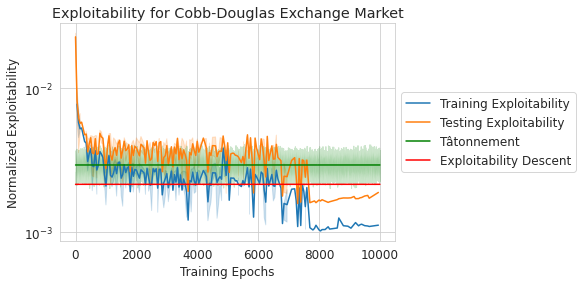}
  \end{subfigure}
  \begin{subfigure}{0.45\textwidth}
    \includegraphics[width=\linewidth]{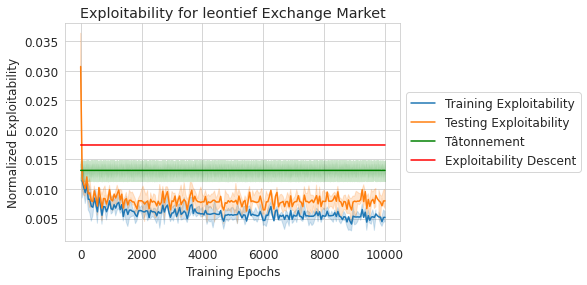}
  \end{subfigure}
  \begin{subfigure}{0.45\textwidth}
    \includegraphics[width=\linewidth]{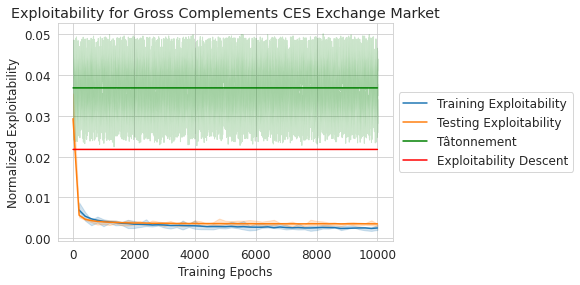}
  \end{subfigure}
  \begin{subfigure}{0.45\textwidth}
    \includegraphics[width=\linewidth]{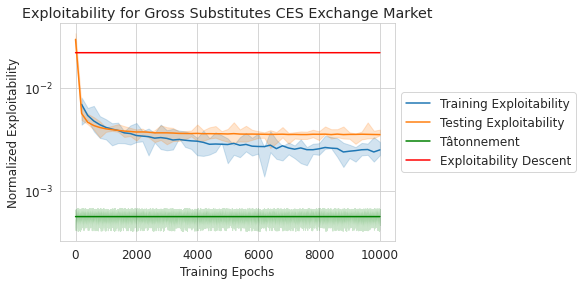}
  \end{subfigure}
  \begin{subfigure}{0.45\textwidth}
    \includegraphics[width=\linewidth]{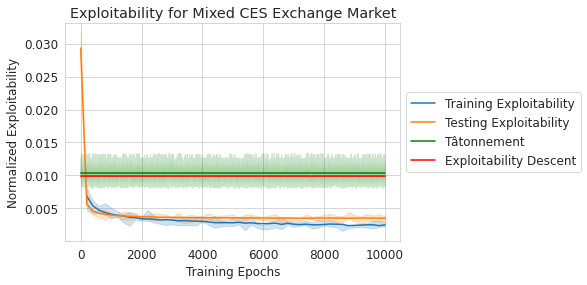}
  \end{subfigure}
  \caption{Training and Testing Exploitability of \nees{} in linear, Cobb-Douglas, Leontief, GS CES, GC CES, and mixed CES economies.
  \label{fig:arrow_exploit}}
\end{figure*}

\begin{figure*}[!ht]
  \centering
  \begin{subfigure}{0.45\linewidth}
    \includegraphics[width= \textwidth]{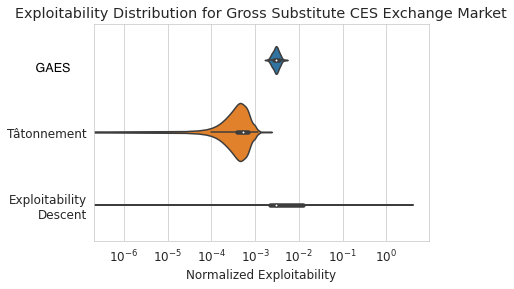}
  \end{subfigure}
  \begin{subfigure}{0.45\textwidth}
    \includegraphics[width=\linewidth]{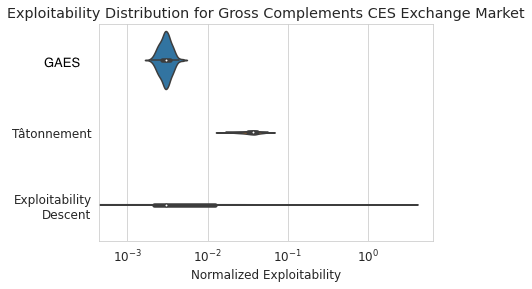}
  \end{subfigure}
  \caption{Distribution of test exploitability on exchange economies/pseudo-games.
   \nees{} outperforms all baselines on average in all markets, and in distribution in all markets except Cobb-Douglas.
  \label{fig:violin_ad_ap}}
\end{figure*}

%% file: kyoto.tex
\subsection{Kyoto Joint Implementation Mechanism}

\paragraph{Experimental Setup}
For this experiment, we focus on computing a refinement of the GNE known as VE (see \Cref{sec_ap:prelims}),
which are guaranteed to exist for this jointly convex pseudo-games.\footnote{We recall that a mydef{variational equilibrium (VE)} (or \mydef{normalized GNE}) of a pseudo-game is a strategy profile $\action^* \in \actions$ s.t.\ for all $\player \in \players$ and $\action \in \actions$, $\util[\player](\action^*) \geq  \util[\player](\action[\player], \naction[\player][][][*])$. See \Cref{sec_ap:prelims} for more details.} %
This does not change the structure of the generator of \nees{} or the training algorithm.
However, it allows us to consider discriminators that output best-responses that are in the space of jointly feasible actions rather than in the space of individually feasible action spaces, greatly simplifying the architecture of our discriminator.  We aim to first replicate the results provided in section 4 of \citeauthor{breton2006game} \cite{breton2006game}. To do so, we first consider a 2 country Kyoto JI mechanism, with all parameters of the Kyoto JI mechanism except $\revrate$ fixed, and compute equilibria for different values of $\revrate$ (\Cref{fig:phase_kyoto}). We then also consider a 2 country Kyoto JI mechanism where we sample all parameters randomly (\Cref{fig:exploit_kyoto}).

\paragraph{Pseudo-Game Generation.}
We sample 12,000 pseudo-games, putting
aside 1,000 for validation and 1,000 for testing. 
We sample the payoff and constraint parameters  of all the players ($\revrate$, $\investrate$, $\emitcap$, $\damagerate$), uniformly in the range $[0.5, 50]$ to produce the pseudo-games.
For each of these pseudo-games, since the set of jointly feasible actions is a polytope, we also generate the vertices associated with the set of jointly feasible actions. To do so, we use the pycdd library \cite{pycddlib2015}, and store a matrix of vertices for each pseudo-game, where the rows correspond to the maximum number of vertices, denoted $\mathrm{MaxNumVertex}$, for any pseudo-game in the training set, and the columns correspond to the dimension of the action space, i.e., $\numcountries*(\numcountries + 1)$ (the first row corresponds to emissions, the  last $\numcountries$ rows correspond to the investment matrix).
For the experiments, and replicating the comparative static analysis of \citet{breton2006game}, 
we randomly sample and fix all parameters of the game except $\revrate$, and
sample $\revrate$ from the range $[0.5, 50]^\numcountries$ as stated above. 

\paragraph{Hyperparameters.}
We run our algorithm with an initial warm-up of 10,000 iterations for the discriminator. This warm-up procedure follows exactly the inner loop of \Cref{alg:cumul_regret} but instead uses randomly sampled pseudo-games and randomly sampled action profiles. After the warmup, we use only one step of inner loop iteration for running \Cref{alg:cumul_regret} and run the outer loop for 10,000 iterations. Together with the warmup, the small number of inner loop iterations allow us to significantly speed up the training process. The gradient step in our algorithm is a step of the ADAM algorithm. For the discriminator, we use the same learning rate for the warm-up and regular training, following a grid search which is a learning rate of $0.0001$, while for the generator we use a learning rate of $0.001$. 

\paragraph{Generator Architecture.}

\begin{figure}[!ht]
    \centering
    \includegraphics[width=\linewidth]{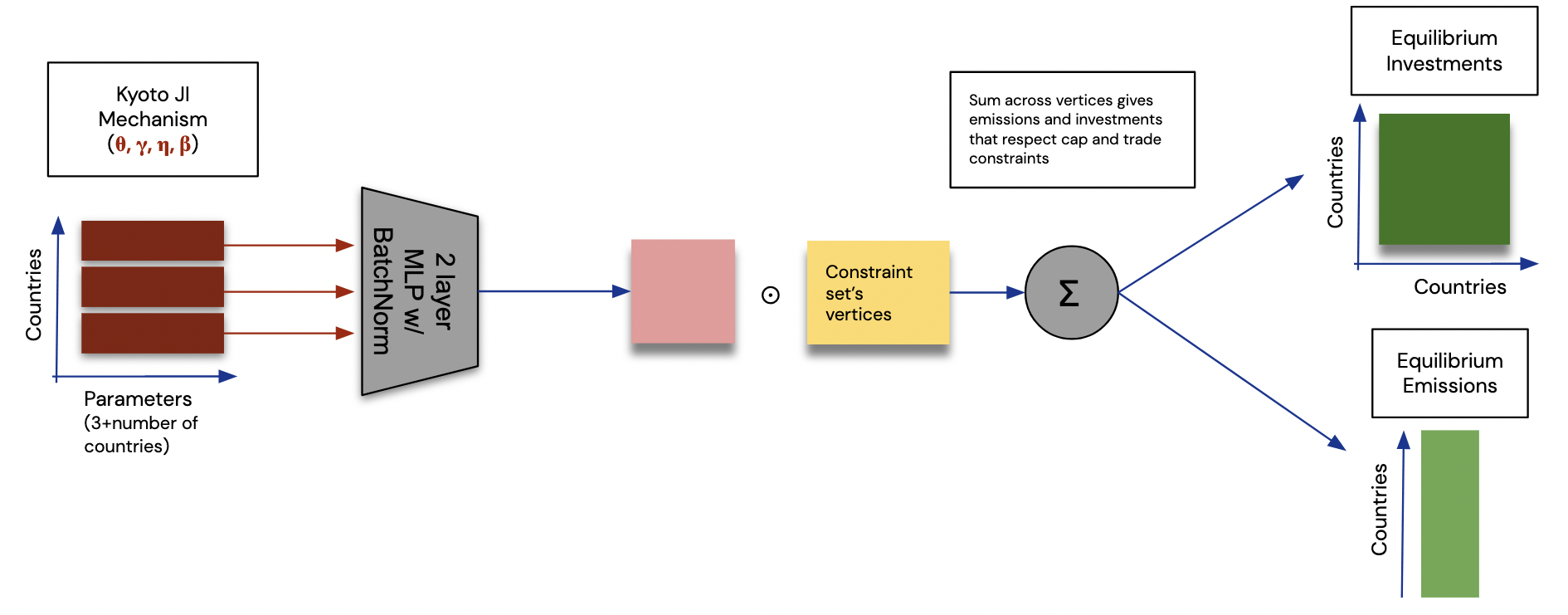}
    \caption{Architecture of the generator of \nees{} for Kyoto JI mechanisms.}
    \label{fig:kyoto_gen}
\end{figure}

We summarize the architecture of \nees{}'s generator for Kyoto JI mechanisms in \Cref{fig:kyoto_gen}. The generator for the Kyoto setting
takes as input the game matrix ($\revrate$, $\investrate$, $\emitcap$, $\damagerate$),
and feeds these inputs through a neural network with two hidden layers, each
with 20 and 30 nodes respectively and ReLU activations.
The output of each layer is also passed through a ReLU
activation, as well as a BatchNorm \cite{ioffe2015batch} layer with trainable parameters, 
and with default parameters as implemented by Haiku.
The output layer of the neural network consists of a fully connected layer with softmax activation with output dimension equal to $\mathrm{MaxNumVertex}$. The output of this final layer is  multiplied with the matrix of vertices associated with the pseudo-game ($\revrate$, $\investrate$, $\emitcap$, $\damagerate$) across its rows, i.e., each vertex associated with pseudo-game's constraint is multiplied by some probability.
The obtained matrix is then summed up across the rows and output by the generator after setting the first column to be $\emit$, and the matrix formed by the remaining columns to be $\invest$. Since the neural network outputs a convex combination of the vertices associated with the constraints of the game, the action profile outputted by the neural network is always jointly feasible. The number of parameters for the generator is 2,824.

\paragraph{Discriminator Architecture}

\begin{figure}[!ht]
    \centering
    \includegraphics[width=\linewidth]{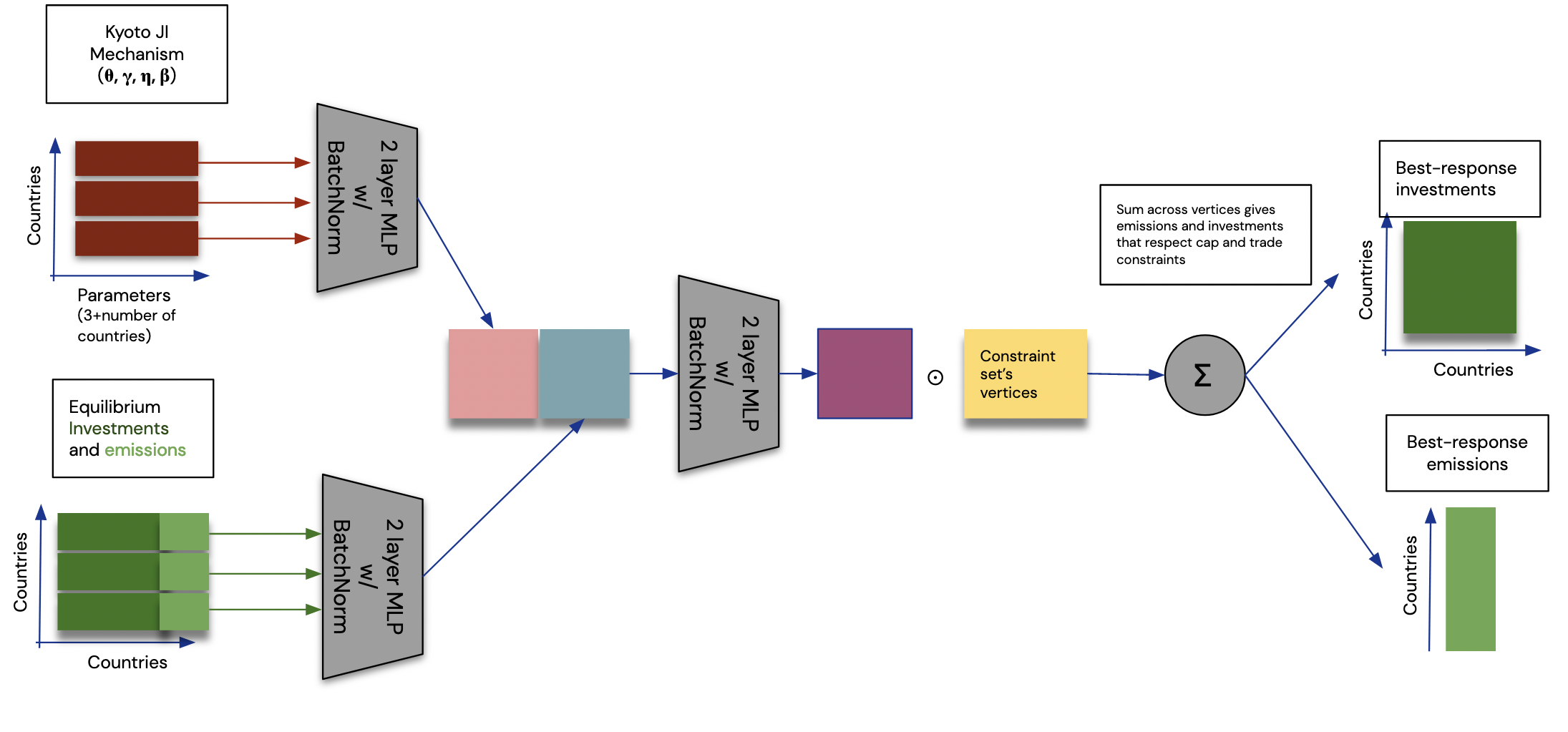}
    \caption{Architecture of the generator of \nees{} for Kyoto JI mechanisms.}
    \label{fig:kyoto_discrim}
\end{figure}

We summarize the architecture of  \nees{}'s discriminator for Kyoto JI mechanisms in \Cref{fig:kyoto_discrim}.Our discriminator takes as input the matrix $(\emit, \invest)$, the output of the generator,
and the pseudo-game matrix ($\revrate$, $\investrate$, $\emitcap$, $\damagerate$).
The equilibrium $(\emit, \invest)$ is first passed through a neural network with two fully connected 
trainable BatchNorm layers, each with 500 nodes.
Similarly, the pseudo-game ($\revrate$, $\investrate$, $\emitcap$, $\damagerate$) is passed through a network with the same architecture. The output of both networks are then concatenated over the last dimension, i.e., the matrices outputted by both networks are appended horizontally so as to preserve the number of rows $\numcountries$.
and passed through a neural network with two fully connected trainable BatchNorm layers, each
with 500 nodes.
For each of these layers, the output  is passed through a ReLU
activation. The output of the last neural network is then flattened and passed through a final fully connected layer with a softmax activation. The output of this final layer is  multiplied with the matrix ($\revrate$, $\investrate$, $\emitcap$, $\damagerate$) of vertices associated with the pseudo-game across its rows.
The obtained matrix is then summed up across the rows and output by the discriminator after setting the first column to be $\emit$, and with the matrix formed by the remaining columns of the output adopted as 
$\invest$. Since the neural network outputs a convex combination of the vertices associated with the constraints of the game, the action profile is always jointly feasible, meaning that the neural network outputs a best-response (for a VE). The number of parameters for the discriminator is 1,302,544.

\begin{figure}[!ht]
    \centering
    \includegraphics[width=0.5\linewidth]{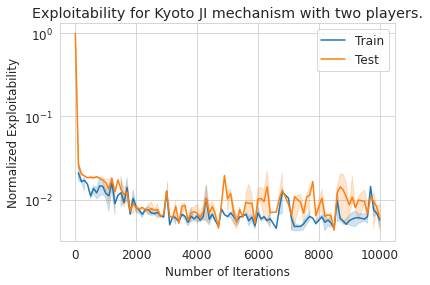}
    \caption{Normalized exploitability achieved by \nees{} throughout training for a two country JI mechanism.
    \label{fig:exploit_kyoto}}
\end{figure}